\RequirePackage{afterpackage}
\AfterPackage{amsthm}{
  \RequirePackage{hyperref,cleveref}
}

\documentclass[final]{lipics-v2021}
\pdfoutput=1

\hideLIPIcs

\nolinenumbers

\usepackage{lscape}
\usepackage{subcaption}
\usepackage{scalerel}
\usepackage{mathrsfs}
\usepackage{stmaryrd}
\usepackage{hyphenat}
\usepackage{stmaryrd}
\usepackage{mathpartir}
\usepackage{tikz-cd}
\usepackage[inline]{enumitem}
\usepackage{mathtools}
\usepackage{wrapfig}

\usepackage[leftmargin=1em,rightmargin=1ex,vskip=1ex]{quoting}
\usepackage[utf8]{inputenc}
\usepackage{cmll}

\newcommand\scalemath[3]{\scalebox{#1}[#2]{\mbox{\ensuremath{\displaystyle #3}}}}

\newcommand{\leftarrowtip}{\ensuremath{\tikz\draw[line width=0.5pt,->] (10pt,0) -- (0,0);}}
\newcommand{\leftarrowtailnotip}{\ensuremath{\tikz\draw[line width=0.5pt,-<] (0,0) -- (10pt,0);}}

\newcommand{\unicodeStar}{\ensuremath{\star}}
\DeclareUnicodeCharacter{2605}{\unicodeStar}

\DeclareUnicodeCharacter{21D2}{\ensuremath{\Rightarrow}}
\DeclareUnicodeCharacter{2218}{\ensuremath{\circ}}
\DeclareUnicodeCharacter{2022}{\ensuremath{\bullet}}
\DeclareUnicodeCharacter{2219}{\ensuremath{\bullet}}
\DeclareUnicodeCharacter{2026}{\ensuremath{\dots}}
\DeclareUnicodeCharacter{2208}{\ensuremath{\in}}
\DeclareUnicodeCharacter{2192}{\ensuremath{\to}}
\DeclareUnicodeCharacter{2190}{\ensuremath{\leftarrowtip}}
\DeclareUnicodeCharacter{2919}{\ensuremath{\leftarrowtailnotip}}

\DeclareUnicodeCharacter{00D7}{\ensuremath{\times}}
\DeclareUnicodeCharacter{00B7}{\ensuremath{\cdot}}
\DeclareUnicodeCharacter{222B}{\ensuremath{\int}}
\DeclareUnicodeCharacter{22A4}{\ensuremath{\top}}
\DeclareUnicodeCharacter{22A5}{\ensuremath{\bot}}
\DeclareUnicodeCharacter{2264}{\ensuremath{\leq}}

\newcommand{\unicodecolon}{\ensuremath{\colon}}
\DeclareUnicodeCharacter{FE55}{\unicodecolon}
\newcommand{\unicodeleftpar}{\ensuremath{\left(}}
\DeclareUnicodeCharacter{27EE}{\unicodeleftpar}
\newcommand{\unicoderightpar}{\ensuremath{\right)}}
\DeclareUnicodeCharacter{27EF}{\unicoderightpar}
\DeclareUnicodeCharacter{2260}{\neq}
\DeclareUnicodeCharacter{22A9}{\Vdash}
\DeclareUnicodeCharacter{2237}{\proportion}
\DeclareUnicodeCharacter{2124}{\mathbb{Z}}
\DeclareUnicodeCharacter{27E8}{\langle}
\DeclareUnicodeCharacter{27E9}{\rangle}
\DeclareUnicodeCharacter{21A6}{\mapsto}
\DeclareUnicodeCharacter{22A2}{\vdash}
\DeclareUnicodeCharacter{2090}{\ensuremath{{}_a}}
\DeclareUnicodeCharacter{A71B}{{}^\uparrow}
\DeclareUnicodeCharacter{A71C}{{}^\downarrow}
\DeclareUnicodeCharacter{27E6}{\llbracket}
\DeclareUnicodeCharacter{27E7}{\rrbracket}

\newcommand{\unicoderightcircle}{\ensuremath{\RIGHTcircle}}
\DeclareUnicodeCharacter{25D1}{\unicoderightcircle}
\newcommand{\unicodeleftcircle}{\ensuremath{\LEFTcircle}}
\DeclareUnicodeCharacter{25D0}{\unicodeleftcircle}
\DeclareUnicodeCharacter{229B}{\circledast}

\newcommand{\unicodebbA}{\ensuremath{\mathbb{A}}}
\DeclareUnicodeCharacter{1D538}{\unicodebbA}
\newcommand{\unicodebbB}{\ensuremath{\mathbb{B}}}
\DeclareUnicodeCharacter{1D539}{\unicodebbB}
\newcommand{\unicodebbC}{\ensuremath{\mathbb{C}}}
\DeclareUnicodeCharacter{2102}{\unicodebbC}
\DeclareUnicodeCharacter{1D53B}{\ensuremath{\mathbb{D}}}
\DeclareUnicodeCharacter{2115}{\ensuremath{\mathbb{N}}}
\DeclareUnicodeCharacter{211D}{\ensuremath{\mathbb{R}}}
\DeclareUnicodeCharacter{1D543}{\ensuremath{\mathbb{L}}}
\newcommand\UnicodeBlackboardP{\ensuremath{\mathbf{P}}} \DeclareUnicodeCharacter{2119}{\UnicodeBlackboardP}
\DeclareUnicodeCharacter{211A}{\ensuremath{\mathbb{Q}}}
\DeclareUnicodeCharacter{1D544}{\ensuremath{\mathbb{M}}}
\DeclareUnicodeCharacter{1D54C}{\ensuremath{\mathbb{U}}}
\DeclareUnicodeCharacter{1D54D}{\ensuremath{\mathbf{V}}}
\DeclareUnicodeCharacter{1D54E}{\ensuremath{\mathbb{W}}}
\DeclareUnicodeCharacter{1D546}{\ensuremath{\mathbb{O}}}
\DeclareUnicodeCharacter{1D540}{\ensuremath{\mathbb{I}}}
\DeclareUnicodeCharacter{1D54A}{\ensuremath{\mathbb{S}}}
\DeclareUnicodeCharacter{1D53C}{\ensuremath{\mathbb{E}}}

\newcommand{\unicodecalS}{\ensuremath{\mathcal{S}}}
\newcommand{\unicodecalT}{\ensuremath{\mathcal{T}}}
\newcommand{\unicodecalC}{\ensuremath{\mathcal{C}}}
\newcommand{\unicodecalD}{\ensuremath{\mathcal{D}}}
\newcommand{\unicodecalX}{\ensuremath{\mathcal{X}}}
\newcommand{\unicodecalN}{\ensuremath{\mathcal{N}}}
\newcommand{\unicodecalE}{\ensuremath{\mathcal{E}}}
\DeclareUnicodeCharacter{1D4D4}{\unicodecalE}
\DeclareUnicodeCharacter{1D4D2}{\unicodecalC}
\DeclareUnicodeCharacter{1D4D3}{\unicodecalD}
\DeclareUnicodeCharacter{1D4DE}{\mathcal{O}}
\DeclareUnicodeCharacter{1D4AA}{\mathcal{O}}
\DeclareUnicodeCharacter{210B}{\mathcal{H}}
\DeclareUnicodeCharacter{1D4E2}{\unicodecalS}
\DeclareUnicodeCharacter{1D4E3}{\unicodecalT}
\DeclareUnicodeCharacter{1D4E7}{\unicodecalX}
\DeclareUnicodeCharacter{1D4D0}{\ensuremath{\mathcal{A}}}
\DeclareUnicodeCharacter{1D4D1}{\ensuremath{\mathcal{B}}}
\DeclareUnicodeCharacter{1D4D6}{\ensuremath{\mathcal{G}}}
\DeclareUnicodeCharacter{1D4D7}{\ensuremath{\mathcal{H}}}
\DeclareUnicodeCharacter{1D4DB}{\ensuremath{\mathcal{L}}}
\DeclareUnicodeCharacter{1D4DC}{\ensuremath{\mathcal{M}}}
\DeclareUnicodeCharacter{1D4DD}{\unicodecalN}
\DeclareUnicodeCharacter{1D4B1}{\ensuremath{\mathcal{V}}}
\DeclareUnicodeCharacter{1D4E5}{\ensuremath{\mathcal{V}}}
\DeclareUnicodeCharacter{1D4E6}{\ensuremath{\mathcal{W}}}
\DeclareUnicodeCharacter{1D4E4}{\ensuremath{\mathcal{U}}}

\DeclareUnicodeCharacter{03B1}{\alpha}
\DeclareUnicodeCharacter{03B2}{\beta}
\DeclareUnicodeCharacter{03BC}{\mu}
\DeclareUnicodeCharacter{03B4}{\delta}
\DeclareUnicodeCharacter{03B5}{\varepsilon}
\DeclareUnicodeCharacter{03B7}{\eta}
\DeclareUnicodeCharacter{03BB}{\lambda}
\DeclareUnicodeCharacter{03C1}{\rho}
\DeclareUnicodeCharacter{03C8}{\psi}
\DeclareUnicodeCharacter{03C4}{\tau}
\DeclareUnicodeCharacter{03A8}{\Psi}
\DeclareUnicodeCharacter{03C3}{\sigma}
\DeclareUnicodeCharacter{03C6}{\varphi}
\DeclareUnicodeCharacter{03A6}{\Phi}
\DeclareUnicodeCharacter{03A3}{\Sigma}
\DeclareUnicodeCharacter{03D5}{\phi}
\DeclareUnicodeCharacter{03B8}{\theta}
\DeclareUnicodeCharacter{03C0}{\ensuremath{\pi}}
\DeclareUnicodeCharacter{0393}{\Gamma}
\DeclareUnicodeCharacter{0394}{\Delta}
\DeclareUnicodeCharacter{03BA}{\kappa}
\DeclareUnicodeCharacter{03BD}{\nu}
\DeclareUnicodeCharacter{03A9}{\Omega}
\DeclareUnicodeCharacter{03C9}{\omega}
\DeclareUnicodeCharacter{03D2}{\Upsilon}
\DeclareUnicodeCharacter{03C5}{\upsilon}

\DeclareUnicodeCharacter{25A0}{\blacksquare}
\DeclareUnicodeCharacter{25AA}{\blacksquare}
\DeclareUnicodeCharacter{2205}{\varnothing}
\DeclareUnicodeCharacter{2254}{\coloneq}
\DeclareUnicodeCharacter{03B3}{\gamma}

\newcommand{\hirayo}{\scaleobj{0.9}{\text{\usefont{U}{min}{m}{n}\symbol{'210}}}}
\DeclareUnicodeCharacter{3088}{\hirayo}
\DeclareFontFamily{U}{min}{}
\DeclareFontShape{U}{min}{m}{n}{<-> udmj30}{}

\newcommand\UnicodeWhiteRightPointingSmallTriangle{\triangleright}
\DeclareUnicodeCharacter{25B9}{\mathbin{\UnicodeWhiteRightPointingSmallTriangle}}
\newcommand\UnicodeWhiteDownPointingSmallTriangle{\triangledown}
\DeclareUnicodeCharacter{25BF}{\mathbin{\UnicodeWhiteDownPointingSmallTriangle}}
\newcommand\UnicodeWhiteUpPointingSmallTriangle{\scalemath{1}{-1}{{}^{\triangledown}}}
\DeclareUnicodeCharacter{25B5}{\mathbin{\UnicodeWhiteUpPointingSmallTriangle}}

\DeclareUnicodeCharacter{2080}{\ensuremath{{}_0}}
\DeclareUnicodeCharacter{2081}{\ensuremath{{}_1}}
\DeclareUnicodeCharacter{2082}{\ensuremath{{}_2}}
\DeclareUnicodeCharacter{2083}{\ensuremath{{}_3}}

\DeclareUnicodeCharacter{1D62}{\ensuremath{{}_i}}
\DeclareUnicodeCharacter{2C7C}{\ensuremath{{}_j}}
\DeclareUnicodeCharacter{02B3}{\ensuremath{{}^r}}
\DeclareUnicodeCharacter{02E1}{\ensuremath{{}^\ell}}
\DeclareUnicodeCharacter{1D48}{\ensuremath{{}^d}}
\DeclareUnicodeCharacter{1D50}{\ensuremath{{}^m}}
\DeclareUnicodeCharacter{1D58}{\ensuremath{{}^u}}
\DeclareUnicodeCharacter{209A}{\ensuremath{{}_p}}
\DeclareUnicodeCharacter{2096}{\ensuremath{{}_k}}

\DeclareUnicodeCharacter{2245}{\ensuremath{\cong}}
\DeclareUnicodeCharacter{2286}{\subseteq}

\DeclareUnicodeCharacter{22C5}{\cdot}
\DeclareUnicodeCharacter{25C3}{\ensuremath{\triangleleft}}
\DeclareUnicodeCharacter{25B9}{\ensuremath{\triangleright}}

\DeclareUnicodeCharacter{2217}{\ast}
\DeclareUnicodeCharacter{039E}{\Xi}
\DeclareUnicodeCharacter{2295}{\oplus}
\DeclareUnicodeCharacter{2297}{\otimes}
\DeclareUnicodeCharacter{214B}{\parr}
\DeclareUnicodeCharacter{2298}{\oslash}
\DeclareUnicodeCharacter{25C0}{\mathbin{\blacktriangleleft}}
\DeclareUnicodeCharacter{25C1}{\mathbin{\vartriangleleft}}
\DeclareUnicodeCharacter{22B3}{\mathbin{\triangleright}}
\DeclareUnicodeCharacter{22B2}{\mathbin{\triangleleft}}
\DeclareUnicodeCharacter{FF5C}{\mid}
\DeclareUnicodeCharacter{227A}{\mathbin{\prec}}
\DeclareUnicodeCharacter{227B}{\mathbin{\succ}}
\DeclareUnicodeCharacter{22A3}{\mathbin{\dashv}}
\DeclareUnicodeCharacter{219D}{\ensuremath{\leadsto}}
\DeclareUnicodeCharacter{2191}{\ensuremath{\uparrow}}
\DeclareUnicodeCharacter{1361}{\colon}

\DeclareUnicodeCharacter{29D1}{\mathrel{\multimapdotbothB}}
\DeclareUnicodeCharacter{29D2}{\mathrel{\multimapdotbothA}}
\DeclareUnicodeCharacter{22C4}{\mathbin{\diamond}}
\DeclareUnicodeCharacter{226B}{\mathrel{\gg}}
\DeclareUnicodeCharacter{25A1}{\mathbin{\square}}
\DeclareUnicodeCharacter{266F}{\sharp}

\DeclareUnicodeCharacter{2113}{\ell}
\DeclareUnicodeCharacter{2261}{\equiv}
\DeclareUnicodeCharacter{2099}{_n}
\DeclareUnicodeCharacter{2098}{_m}
\DeclareUnicodeCharacter{1D5AD}{\ensuremath{\mathsf{N}}}

\DeclareUnicodeCharacter{1D4DF}{\mathcal{P}}
\DeclareUnicodeCharacter{1D415}{\mathbf{V}}
\DeclareUnicodeCharacter{1D400}{\mathbf{A}}
\DeclareUnicodeCharacter{1D6C2}{\pmb{\alpha}}
\DeclareUnicodeCharacter{1D6C3}{\pmb{\beta}}
\DeclareUnicodeCharacter{1D6C4}{\pmb{\gamma}}
\DeclareUnicodeCharacter{1D6C5}{\pmb{\delta}}
\DeclareUnicodeCharacter{1D6C6}{\pmb{\epsilon}}
\DeclareUnicodeCharacter{1D6C8}{\pmb{\eta}}
\DeclareUnicodeCharacter{1D6DA}{\pmb{\omega}}

\newcommand\mydots{\makebox[0.6em][c]{.\hfil.\hfil.}}
\DeclareUnicodeCharacter{2026}{\mydots}
\DeclareUnicodeCharacter{226B}{\gg}

\DeclareUnicodeCharacter{22C9}{\ltimes}
\DeclareUnicodeCharacter{22CA}{\rtimes}

\newcommand{\unicodeRelationalComposition}{\dcomp}
\DeclareUnicodeCharacter{2A3E}{\unicodeRelationalComposition}
\DeclareUnicodeCharacter{2983}{\llparenthesis}
\DeclareUnicodeCharacter{2984}{\rrparenthesis}

\input{pkg-color-links.tex}

\makeatletter
\newcommand{\nicelinktarget}[1]{\Hy@raisedlink{\hypertarget{#1}{}}}
\makeatother

\newcommand\Set{\hyperlink{linkSet}{\mathbf{Set}}}

\usepackage[xcolor,no patch,hyperref,quotation,electronic]{knowledge}

\knowledgeconfigure{notion}
\knowledgestyle{notion}{color=nordnight}
\knowledgestyle{intro notion}{emphasize,color=nordnight}

\knowledge{notion}
| uniformity
| Uniformity
| uniformity axiom
| Uniformity axiom
| uniformly traced distributive multicategory
| uniformly traced distributive multicategories
| Uniformly traced distributive multicategory
| Uniformly traced distributive multicategories

\knowledge{notion}
| term
| terms
| Term
| Terms

\knowledge{notion}
| interchange
| Interchange
| interchange axiom
| Interchange axiom

\knowledge{notion}
| correctness triple
| Correctness triple
| correctness triples
| Correctness triples

\knowledge{notion}
| posetal imperative category
| Posetal imperative category
| posetal imperative categories
| Posetal imperative categories

\knowledge{notion}
| posetal imperative multicategory
| Posetal imperative multicategory
| posetal imperative multicategories
| Posetal imperative multicategories

\knowledge{notion}
| guard combinators
| Guard combinators
| guard combinator
| Guard combinator

\knowledge{notion}
| command combinators
| Command combinators
| command combinator
| Command combinator

\knowledge{notion}
| monoidal category
| monoidal categories
| Monoidal category
| Monoidal categories

\knowledge{notion}
| sesquifunctor
| Sesquifunctor
| sesquifunctors
| Sesquifunctors

\knowledge{notion}
| premonoidal category
| Premonoidal category
| premonoidal categories
| Premonoidal categories
| premonoidality
| Premonoidality
| premonoidal
| Premonoidal

\knowledge{notion}
| central
| Central
| central morphism
| Central morphism
| central morphisms
| Central morphisms

\knowledge{notion}
| Symmetric premonoidal category
| Symmetric premonoidal categories
| symmetric premonoidal category
| symmetric premonoidal categories

\knowledge{notion}
| deterministic
| Deterministic
| deterministic morphism
| Deterministic morphism
| deterministic morphisms
| Deterministic morphisms

\knowledge{notion}
| total
| total morphism
| total morphisms
| Total
| Total morphism
| Total morphisms

\knowledge{notion}
| cocartesian multicategory
| Cocartesian multicategory
| cocartesian multicategories
| Cocartesian multicategories

\knowledge{notion}
| predistributive multicategory
| Predistributive multicategory
| predistributive multicategories
| Predistributive multicategories

\knowledge{notion}
| posetal distributive copy-discard multicategory
| Posetal distributive copy-discard multicategory
| posetal distributive copy-discard multicategories
| Posetal distributive copy-discard multicategories

\knowledge{notion}
| index
| indices
| Index
| Indices

\knowledge{notion}
| state
| states
| State
| States

\knowledge{notion}
| state combinator
| state combinators
| State combinator
| State combinators

\knowledge{notion}
| multicategory
| Multicategory
| Multicategories
| multicategories

\knowledge{notion}
| distributive multicategory
| distributive multicategories
| Distributive multicategory
| Distributive multicategories

\knowledge{notion}
| substitution
| Substitution
| substitutions
| Substitutions
| substituting
| Substituting

\knowledge{notion}
| guard
| guards
| Guard
| Guards

\knowledge{notion}
| predicate
| predicates
| Predicates
| Predicate

\knowledge{notion}
| predicate combinators
| Predicate combinators

\knowledge{notion}
| command
| commands
| Commands
| Command

\knowledge{notion}
| fresh
| Fresh

\knowledge{notion}
| alpha-equivalent
| alpha-equivalence
| Alpha-equivalent
| Alpha-equivalence

\knowledge{notion}
| variable
| variables
| Variable
| Variables

\knowledge{notion}
| traced distributive copy-discard multicategory
| Traced distributive copy-discard multicategory
| traced distributive copy-discard multicategories
| Traced distributive copy-discard multicategories

\knowledge{notion}
| traced distributive multicategory
| Traced distributive multicategory
| traced distributive multicategories
| Traced distributive multicategories

\knowledge{notion}
| traced predistributive copy-discard multicategory
| Traced predistributive copy-discard multicategory
| traced predistributive copy-discard multicategories
| Traced predistributive copy-discard multicategories

\knowledge{notion}
| predistributive copy-discard multicategory
| Predistributive copy-discard multicategory
| predistributive copy-discard multicategories
| Predistributive copy-discard multicategories

\knowledge{notion}
| distributive copy-discard multicategory
| Distributive copy-discard multicategory
| distributive copy-discard multicategories
| Distributive copy-discard multicategories

\knowledge{notion}
| variable substitution
| variable substitutions
| Variable substitution
| Variable substitutions

\knowledge{notion}
| posetal distributive signature
| posetal distributive signatures
| Posetal distributive signature
| Posetal distributive signatures

\knowledge{notion}
| posetal uniformity
| Posetal uniformity

\knowledge{notion}
| posetal distributive copy-discard category
| posetal distributive copy-discard categories
| Posetal distributive copy-discard category
| Posetal distributive copy-discard categories

\knowledge{notion}
| posetal uniform trace
| posetal uniform traces
| Posetal uniform traces
| Posetal uniform trace
| posetal uniform traced monoidal category
| Posetal uniform traced monoidal category
| posetal uniform traced monoidal categories
| Posetal uniform traced monoidal categories

\knowledge{notion}
| label substitution
| label substitutions
| Label substitution
| Label substitutions

\knowledge{notion}
| anchor
| anchors
| Anchor
| Anchors
| label
| Label
| Labels
| labels

\knowledge{notion}
| context
| Context
| contexts
| Contexts

\knowledge{notion}
| copy-discard category
| Copy-discard category
| copy-discard categories
| Copy-discard categories
| copy-discard
| Copy-discard

\knowledge{notion}
| copy-discard monoidal category
| Copy-discard monoidal category
| copy-discard monoidal categories
| Copy-discard monoidal categories

\knowledge{notion}
| copy-discard premonoidal category
| Copy-discard premonoidal category
| copy-discard premonoidal categories
| Copy-discard premonoidal categories
| premonoidal copy-discard category
| Premonoidal copy-discard category
| premonoidal copy-discard categories
| Premonoidal copy-discard categories

\knowledge{notion}
| commutative imperative category
| Commutative imperative category
| commutative imperative categories
| Commutative imperative categories

\knowledge{notion}
| imperative category
| Imperative category
| imperative categories
| Imperative categories

\knowledge{notion}
| imperative multicategory
| Imperative multicategory
| imperative multicategories
| Imperative multicategories

\knowledge{notion}
| bimonoidally strict
| Bimonoidally strict
| strict
| Strict

\knowledge{notion}
| distributive signature
| distributive signatures
| Distributive signature
| Distributive signatures

\knowledge{notion}
| distributive category
| distributive categories
| Distributive category
| Distributive categories
| distributive copy-discard category
| distributive copy-discard categories
| Distributive copy-discard category
| Distributive copy-discard categories

\knowledge{notion}
| basic type
| basic types
| Basic type
| Basic types

\knowledge{notion}
| generator
| generators
| Generator
| Generators

\knowledge{notion}
| effectful triple
| effectful triples
| Effectful triple
| Effectful triples

\knowledge{notion}
| double signature
| double signatures
| Double signature
| Double signatures

\knowledge{notion}
| forgetful double signature
| forgetful double signatures
| Forgetful double signature
| Forgetful double signatures

\knowledge{notion}
| double signature morphism
| double signature morphisms
| Double signature morphism
| Double signature morphisms

\knowledge{notion}
| pinwheel double category
| pinwheel double categories
| Pinwheel double category
| Pinwheel double categories
| double category with pinwheels
| double categories with pinwheels
| Double category with pinwheels
| Double categories with pinwheels

\knowledge{notion}
| bigraph
| Bigraph
| bigraphs
| Bigraphs
| 2-graph
| 2-graphs
| 2-Graph
| 2-Graphs

\knowledge{notion}
| 2-graph morphism
| 2-graph morphisms

\knowledge{notion}
| 2-graph category
| 2-graph categories

\knowledge{notion}
| double category
| double categories
| Double category
| Double categories

\knowledge{notion}
| weighted polygraph
| weighted polygraphs
| Weighted polygraph
| Weighted polygraphs
| timed polygraph
| timed polygraphs
| Timed polygraph
| Timed polygraphs

\knowledge{notion}
| morphism of timed polygraphs
| Morphism of timed polygraphs
| timed polygraph morphism
| Timed polygraph morphism

\knowledge{notion}
| tilted bicategory signature
| tilted bicategory signatures
| Tilted bicategory signature
| Tilted bicategory signatures
| tilted bigraph
| Tilted bigraph
| tilted bigraphs
| Tilted bigraphs
| tilted 2-graph
| Tilted 2-graph
| tilted 2-graphs
| Tilted 2-graphs

\knowledge{notion}
| path
| paths
| Path
| Paths
| composable path
| composable paths

\knowledge{notion}
| signature
| signatures
| Signature
| Signatures

\knowledge{notion}
| congruence
| congruences
| Congruence
| Congruences

\knowledge{notion}
| symmetry
| symmetries
| Symmetry
| Symmetries

\knowledge{notion}
| profunctor
| profunctors
| Profunctor
| Profunctors

\knowledge{notion}
| left inclusion
| Left inclusion

\knowledge{notion}
| right inclusion
| Right inclusion

\knowledge{notion}
| left whiskering
| Left whiskering

\knowledge{notion}
| right whiskering
| Right whiskering

\knowledge{notion}
| inclusion and whiskering
| Inclusion and whiskering
| whiskering interacts with inclusion
| Whiskering interacts with inclusion

\knowledge{notion}
| whiskering
| Whiskering

\knowledge{notion}
| rewiring preserves whiskering
| Rewiring preserves whiskering
| whiskering preserves rewiring
| Whiskering preserves rewiring

\knowledge{notion}
| rewiring preserves interchange
| Rewiring preserves interchange

\knowledge{notion}
| rewiring preserves composition
| Rewiring preserves composition

\knowledge{notion}
| rewiring
| Rewiring 
| definition of rewiring

\knowledge{rewiring is an action}[]{notion}

\knowledge{notion}
| composition
| Composition
| term composition
| Term composition

\knowledge{notion}
| term composition is associative
| Term composition is associative

\knowledge{notion}
| term composition is unital
| Term composition is unital

\knowledge{notion}
| tensoring
| Tensoring

\knowledge{notion}
| arrow notation preterm
| Arrow notation preterm
| arrow notation preterms
| Arrow notation preterms

\knowledge{notion}
| arrow notation term
| Arrow notation term
| arrow notation terms
| Arrow notation terms

\knowledge{notion}
| observe-arrow notation
| Observe-arrow notation
| observe-arrow notation term
| Observe-arrow notation term
| observe-arrow notation terms
| Observe-arrow notation terms

\knowledge{notion}
| subdistribution
| Subdistribution
| subdistributions
| Subdistributions
| subdistributional

\knowledge{notion}
| rescaling
| Rescaling

\knowledge{notion}
| normalisation
| normalisations
| Normalisation
| Normalisations
| normalization
| normalizations
| Normalization
| Normalizations
| normalised
| Normalised

\knowledge{notion}
| observe-arrow notation copy-discard-compare category of terms

\knowledge{notion}
| category of arrow notation terms

\knowledge{notion}
| copy-discard-compare categories
| copy-discard-compare category
| Copy-discard-compare categories
| Copy-discard-compare category

\knowledge{notion}
| Kleisli extension
| Kleisli extensions

\knowledge{notion}
| Kleisli category
| Kleisli categories

\newcommand{\kleisli}[1]{\kl[Kleisli category]{\mathsf{kl}}(#1)}

\knowledge{notion}
| partially additive
| partially additive monad
| Partially additive
| Partially additive monad

\knowledge{notion}
| sets
| category of sets

\renewcommand{\Set}{\kl[sets]{\mathsf{Set}}}

\knowledge{notion}
| powerset
| powerset monad

\knowledge{notion}
| category of relations
| Category of relations

\knowledge{notion}
| maybe
| maybe monad

\knowledge{notion}
| category of partial functions
| Category of partial functions

\knowledge{notion}
| discrete subdistribution
| discrete subdistributions
| discrete subdistribution monad
| finitely-supported subdistribution
| finitely-supported subdistributions
| finitely-supported subdistribution monad

\newcommand{\subdistr}{\kl[discrete subdistribution]{\mathcal{D}_{\leq}}}

\knowledge{notion}
| discrete distribution
| discrete distributions
| discrete distribution monad
| finitely-supported distribution
| finitely-supported distributions
| finitely-supported distribution monad

\newcommand{\distr}{\kl[discrete distribution]{\mathcal{D}}}

\knowledge{notion}
| category of discrete stochastic channels
| Category of discrete stochastic channels

\knowledge{notion}
| support

\knowledge{notion}
| Dirac distribution

\knowledge{notion}
| measurable subdistribution
| measurable subdistributions
| measurable subdistribution monad
| Panangaden monad
| subprobability measure monad

\newcommand{\subgiry}{\kl[measurable subdistribution]{\mathcal{G}_{\leq}}}

\knowledge{notion}
| measurable distribution
| measurable distributions
| measurable distribution monad
| Giry monad
| probability measure monad

\newcommand{\giry}{\kl[Giry monad]{\mathcal{G}}}

\knowledge{notion}
| standard Borel spaces
| Standard Borel spaces

\knowledge{notion}
| quantale-valued sets

\knowledge{notion}
| assertion-correctness triple
| assertion-correctness triples
| Assertion-correctness triple
| Assertion-correctness triples

\knowledge{notion}
| state-correctness triple
| state-correctness triples
| State-correctness triple
| State-correctness triples

\knowledge{notion}
| predicate-correctness triple
| predicate-correctness triples
| Predicate-correctness triple
| Predicate-correctness triples

\knowledge{notion}
| assertion-incorrectness triple
| assertion-incorrectness triples
| Assertion-incorrectness triple
| Assertion-incorrectness triples

\knowledge{notion}
| state-incorrectness triple
| state-incorrectness triples
| State-incorrectness triple
| State-incorrectness triples

\knowledge{notion}
| predicate-incorrectness triple
| predicate-incorrectness triples
| Predicate-incorrectness triple
| Predicate-incorrectness triples

\knowledge{notion}
| relational assertion-correctness triple
| relational assertion-correctness triples
| Relational assertion-correctness triple
| Relational assertion-correctness triples

\knowledge{notion}
| relational state-correctness triple
| relational state-correctness triples
| Relational state-correctness triple
| Relational state-correctness triples

\knowledge{notion}
| relational predicate-correctness triple
| relational predicate-correctness triples
| Relational predicate-correctness triple
| Relational predicate-correctness triples

\knowledge{notion}
| relational assertion-incorrectness triple
| relational assertion-incorrectness triples
| Relational assertion-incorrectness triple
| Relational assertion-incorrectness triples

\knowledge{notion}
| relational state-incorrectness triple
| relational state-incorrectness triples
| Relational state-incorrectness triple
| Relational state-incorrectness triples

\knowledge{notion}
| relational predicate-incorrectness triple
| relational predicate-incorrectness triples
| Relational predicate-incorrectness triple
| Relational predicate-incorrectness triples

\knowledge{notion}
| branch combinator

\knowledge{notion}
| coupling synchronization

\knowledge{notion}
| pll
| parallel composition of natural transformations

\newcommand{\pll}{\mathbin{\kl[pll]{\star}}}

\knowledge{notion}
| monoid of sampling intervals
| sampling intervals monoid
| sampling intervals
| Samp

\newcommand{\Samp}{\kl[Samp]{\mathsf{Samp}}}

\knowledge{notion}
| graded monoid
| graded monoids
| Graded monoid
| Graded monoids

\knowledge{notion}
| coalgebra composition

\newcommand{\ccomp}{\mathbin{\kl[coalgebra composition]{\bullet}}}

\knowledge{notion}
| graded monad morphism
| graded monad morphisms
| Graded monad morphism
| Graded monad morphisms

\IfFileExists{noappendix.token}%
{\usepackage[appendix=strip,bibliography=common]{apxproof}}%
{\usepackage[bibliography=common]{apxproof}}

\newcommand{\cat}[1]{\mathbf{#1}}
\newcommand{\fun}[1]{\mathsf{#1}}

\newcommand{\id}{\mathsf{id}}
\newcommand{\dcomp}{\mathbin{;}}

\newcommand{\tensor}{\otimes}
\newcommand{\cp}{\nu}

\newcommand{\proj}[1][]{\pi_{#1}} %

\newcommand{\naturals}{\mathbb{N}}
\newcommand{\reals}{\mathbb{R}}
\newcommand{\posreals}{\reals_{\geq}}
\newcommand{\integers}{\mathbb{Z}}

\newcommand{\Meas}{\cat{Meas}}

\newcommand{\writer}[1]{\fun{W}_{#1}}

\newcommand{\ket}[1]{|#1\rangle}

\newcommand{\obs}{\mathsf{obs}}

\newcommand{\ar}{\mathsf{ar}}
\newcommand{\depth}{\mathsf{d}}

\newcommand{\modal}[1]{\langle {#1} \rangle}
\newcommand{\rmodal}[1]{( {#1} )}
\newcommand{\GCoAlg}[1]{\cat{GCoAlg}(#1)}

\newcommand{\GPCoAlg}[1]{\cat{GPCoAlg}(#1)}

\definecolor{med0}{HTML}{1C1B1B} %
\definecolor{med1}{HTML}{261D1D} %
\definecolor{med2}{HTML}{362B2B} %
\definecolor{med3}{HTML}{5E5757} %
\definecolor{med4}{HTML}{FFE983} %
\definecolor{med5}{HTML}{FFF4C2} %
\definecolor{med6}{HTML}{FDF6E3} %
\definecolor{med7}{HTML}{DB7842} %
\definecolor{med8}{HTML}{B32E39} %
\definecolor{med9}{HTML}{821529} %
\definecolor{medA}{HTML}{602E51} %
\definecolor{medB}{HTML}{FFC929} %
\definecolor{medC}{HTML}{60C37E} %
\definecolor{medD}{HTML}{89D7D0} %
\definecolor{medE}{HTML}{3E75DA} %
\definecolor{medF}{HTML}{D0ADE1} %
\colorlet{medBlack}{med0}
\colorlet{medWhite}{med6}
\colorlet{medRed}{med8}
\colorlet{medBlue}{medE}
\definecolor{nordred}{HTML}{bf616a}
\definecolor{bordeaux}{HTML}{4b1121}
\definecolor{darkyellow}{HTML}{FFC20A}
\definecolor{nicered}{HTML}{9C0D38}
\definecolor{niceblue}{HTML}{0C7BDC}
\colorlet{localblack}{black}
\colorlet{localwhite}{white}
\colorlet{localcolor}{medB}
\colorlet{localgray}{med3}
\colorlet{localblue}{medE}
\colorlet{localred}{med8}
\colorlet{addcolor}{white}
\colorlet{copycolor}{black}

\definecolor{nordred}{HTML}{bf616a}
\definecolor{bordeaux}{HTML}{821529}
\definecolor{bluelink}{HTML}{003399}
\definecolor{nordred}{HTML}{bf616a}
\definecolor{nordblue}{HTML}{81a1c1}
\definecolor{norddarkblue}{HTML}{5e81ac}
\definecolor{nordgreen}{HTML}{a3be8c}
\definecolor{nordnight}{HTML}{4c566a}

\AtEndPreamble{
  \RequirePackage{hyperref}
  \hypersetup{
    breaklinks=true,
    linktocpage,
    colorlinks = true,
    linkcolor = nordnight,
    citecolor = nordgreen,
    urlcolor = nordblue
  }
\RequirePackage{cleveref}
\crefname{example}{example}{examples}
\Crefname{example}{Example}{Examples}
}

\allowdisplaybreaks{}

\usepackage{cleveref} %

\theoremstyle{plain}
\newtheoremrep{theorem}{Theorem}[section]
\newtheoremrep{proposition}[theorem]{Proposition}
\newtheoremrep{lemma}[theorem]{Lemma}
\newtheoremrep{corollary}[theorem]{Corollary}

\theoremstyle{definition}
\newtheorem{define}[theorem]{Definition}

\theoremstyle{remark}

\crefname{definition}{Definition}{Definitions}
\crefname{theorem}{Theorem}{Theorems}
\crefname{lemma}{Lemma}{Lemmas}
\crefname{proposition}{Proposition}{Propositions}
\crefname{corollary}{Corollary}{Corollaries}
\crefname{example}{example}{examples}
\crefname{remark}{Remark}{Remarks}
\Crefname{example}{Example}{Examples}

\allowdisplaybreaks

\crefname{example}{example}{examples}

\renewcommand{\ldots}{...} %

\tikzcdset{scale cd/.style={
	every label/.append style={scale=#1},
    cells={nodes={scale=#1}}
}}



\title{Graded Monad Coalgebras for Continuous-Time Transition Systems}

\author{Elena Di Lavore}{Talinn University of Technology}{}{https://orcid.org/0000-0002-7783-5079}{}

\author{Jonas Forster}{Friedrich-Alexander-Universität Erlangen-Nürnberg}{}{https://orcid.org/0000-0002-5050-2565}{}

\author{Mario Román}{Talinn University of Technology}{}{https://orcid.org/0000-0003-3158-1226}{}

\date{\today}
\authorrunning{E. Di Lavore, J. Forster, and M. Román}

\makeatletter
\def\@copyrightspace{\relax}
\makeatother

\funding{Jonas Forster was supported by the Deutsche Forschungsgemeinschaft (DFG, German Research Foundation) -- project number 434050016 (SpeQt). Elena Di Lavore and Mario Román were supported by the Estonian Research Council grant PRG 3215, and by the Advanced Research + Invention Agency (ARIA) Safeguarded AI Programme.}

\acknowledgements{The authors would like to thank Stefan Milius for helpful discussions regarding locally presentable categories.}

\begin{document}

\maketitle

\begin{abstract}
  Functor coalgebras capture a wide range of transition systems that must however evolve in discrete steps.
  We introduce graded coalgebras of graded monads and propose them to model continuous-time transition systems.
  We develop the theory of graded coalgebras—including graded distributive laws between graded monads—and we give conditions for the existence of terminal coalgebras. We define both branching-time and trace semantics, linking them to recent work on Feller-Dynkin processes.
  Finally, we develop coalgebraic modal logics for both process semantics and state criteria for invariance and expressivity.
\end{abstract}

\section{Introduction}%
\label{sec:intro}
Coalgebras of functors are an abstraction for state-based systems, in which the transition behaviour is parametric in an endofunctor $F$: they consist of a set of states, \(X\), and a
transition function, \(\gamma \colon X \to FX\)~\cite{rutten2000universal}. For instance, coalgebras for the functor
\(FX = \distr(O \times X)\) are discrete-time hidden Markov models with observations in $O$, if we let \(\distr\) be the finitely-supported distribution monad.
By this definition, coalgebra transitions occur in discrete steps. Could coalgebras also model continuous-time transition systems?

\begin{wrapfigure}[6]{r}{0.40\textwidth}
  \vspace{-0.8cm}
  \begin{equation}%
  \label{eq:transition-graph-rep-system}
  \begin{tikzcd}[column sep=huge, row sep=tiny]
    & L
    \arrow[bend left=35]{rd}[pos=0.6]{\mu}
    \arrow{ld}[pos=0.4]{\lambda}
    &
    \\
    0
    \arrow[bend left=35]{ru}[pos=0.6]{\mu}
    \arrow{rd}[pos=0.6]{\mu}
    &
    & 2
    \arrow{lu}[pos=0.6]{\lambda}
    \arrow[bend left=35]{ld}[pos=0.6]{\lambda}
    \\
    & R
    \arrow{ru}[pos=0.4]{\mu}
    \arrow[bend left=35]{lu}[pos=0.4]{\lambda} &
  \end{tikzcd}
\end{equation}
\end{wrapfigure}
The continuous-time Markov chain in diagram~\eqref{eq:transition-graph-rep-system}, models a two-component repairable system~\cite[Example~6.1]{gnedenko1995probabilistic}.
Its states encode which components are working: In state \(0\), no component is working; in \(L\), only the left one is working; in \(R\), only the right one; and, in \(2\), both of them are working.
Breakages occur independently on each component and are $\lambda$-exponentially distributed, $b(t) = \lambda e^{-\lambda t}$; repairing time is also idependent on each component but $\mu$-exponentially distributed, $r(t) = \mu e^{-\mu t}$.
Assume that a system with one working component, when tested, functions half of the time; i.e., our observation function, \(o \colon \{0,L,R,2\} \to \distr(\{yes, no\})\), is defined as follows.\footnote{We use the ket-notation for probability distributions \cite{jacobs2009hypernormalisation}: A distribution that assigns \(0.5\) probability to both \(yes\) and \(no\) will be written as \(0.5 \ket{yes} + 0.5 \ket{no}.\)}
\begin{align}
  \label{eq:observations-rep-system}
  o(0) &= 1 \ket{no}; & o(2) &= 1 \ket{yes}; & o(L) &= o(R) = 0.5 \ket{yes} + 0.5 \ket{no}.
\end{align}
We have described a hidden Markov model. Let us propose a coalgebraic description.

The system of differential equations for this model has a solution, determined by a family of transition kernels \(\gamma_{t} \colon \{0,L,R,2\} \to \distr(\{0,L,R,2\})\) where \(\gamma_{t}(y \mid x) = \gamma_{t}(x)(y)\) is the probability that, after time \(t\), the system has transitioned from state \(x\) to \(y\).
While it can be computed as a sum of exponentials, \(\gamma_{t}(k \mid j) = \sum_{i = 1}^{4} c_{i,j} v_{k,i} e^{l_{i}t}\) (\Cref{ex:graded-coalgebra-rep-system}), its exact expression is not relevant for the present discussion; 
what is important is that it satisfies the Markov property: i.e., \(\gamma_{0} = \id\) and \(\gamma_{s+t} = \gamma_{s} \dcomp \gamma_{t}\), in the category of finitely-supported stochastic kernels.
We set out to capture these equations by an appropriate definition of \emph{graded coalgebra of a monad}: $\gamma$ is a graded coalgebra of the distribution monad, graded by the monoid of positive reals.

Moreover, its observations, \(o \colon \{0,L,R,2\} \to \distr(\{yes, no\})\), also induce a family of functions
\(\hat{\gamma}_{t,k} \colon \{0,L,R,2\} \to \distr(\{yes, no\}^{k} \times \{0,L,R,2\})\),
that after letting the system transition for \(t \in \posreals\) time, test it $k$ times. Note how multiple tests may return different answers.
More generally, we obtain functions that alternate transitions of \(t_{i}\) time with clusters of \(k_{i}\) observations,
\begin{equation}%
  \label{eq:samp-graded-rep-system}
  \hat{\gamma}_{(t_{0},k_{0}, \dots, t_{n},k_{n})} \colon \{0,L,R,2\} \to \distr(\{yes, no\}^{k} \times \{0,L,R,2\})
\end{equation}
 for \(k = k_{0}+ \cdots +k_{n}\) the total number of observations.

These are coalgebras for the functors \(\distr(\{yes, no\}^{k} \times -)\) satisfying some extra Markov-like property: indeed, $\hat{\gamma}$ is a graded coalgebra of the monad \(\distr(\{yes, no\}^{k} \times -)\), graded by the coproduct of a time monoid, $(\posreals, +, 0)$, and a monoid counting observations, \((\naturals, +, 0)\).

\subparagraph{Graded monad coalgebras.}
We introduce graded coalgebras of graded monads as a mathematical abstraction
for continuous-time transition systems (\Cref{sec:graded-coalgebras}) and continous-time Markov chains in particular.
Graded coalgebras of monads might be surprising: in the non-graded case, coalgebras of monads are trivial. Grades lift this limitation.

Fix a monoid representing time, \((T, \cdot, e)\), and a monad \(M_{t}\) graded by it (\Cref{def:graded-monad}). Graded monad coalgebras, \(\gamma_{t} \colon X \to M_{t}X\), encode transition systems whose transition behaviour at time $t \in T$ is specified by \(\gamma_{t}\). The graded monad coalgebra axioms (\Cref{def:gradedcoalgebra}) impose that instant transitions
produce no observable behaviour—\(\gamma_{e} = \eta\) must be the graded monad
unit—and that a transition of \(s \cdot t\) time yields the same behaviour as a transition of \(s\) time followed by a transition of \(t\) time.

The existing framework of graded semantics~\cite{DorschMS19} uses graded monads and functor coalgebras to capture different types of process equivalence; the individual components of the monad, graded by the monoid \((\naturals, +,0)\), determine the observable behaviours of the system after some number of steps.
For instance, the trace semantics of discrete-time hidden Markov models is captured by the graded monad \(M_{n}X = \distr(B^{n} \times X)\).

We extend this picture to allow examples where time is an arbitrary monoid.
Coalgebras of monads graded by a monoid unify several abstractions of state-based systems, including examples that were out of the scope of coalgebraic models.
Functor coalgebras and graded semantics~\cite{DorschMS19} are instances of graded coalgebras, but also Markov semigroups, Lawvere dynamical systems~\cite{lawvere1984functorial} (\Cref{ex:graded-coalgebras}) and Feller-Dynkin processes~\cite{chen2023behavioural} (\Cref{sec:fd-processes}) are particular graded coalgebras.

\begin{wrapfigure}[4]{r}{0.33\textwidth}
\vspace{-0.7cm}
\begin{equation}%
  \label{eq:bisimilar-rep-system}
  \begin{tikzcd}[baseline=-1cm]
  0
  \arrow[bend left=35]{r}{2\mu}
  & 1
  \arrow[bend left=35]{r}{\mu}
  \arrow[bend left=35]{l}{\lambda}
  & 2
  \arrow[bend left=35]{l}{2\lambda}
  \end{tikzcd}
\end{equation}
\end{wrapfigure}
\subparagraph{Coalgebraic process equivalences.}
The previous Markov chain in Diagram \eqref{eq:transition-graph-rep-system} is not the only model of a two-component repairable system: we may also use a chain with three states, as in Diagram~\eqref{eq:bisimilar-rep-system}~\cite[Section~3.8.5]{shooman2003reliability}.
Both models~(\ref{eq:transition-graph-rep-system},~\ref{eq:bisimilar-rep-system}) are used interchangeably because they are behaviourally equivalent: intuitively, the states \(L\) and \(R\) in model~\eqref{eq:transition-graph-rep-system} are indistinguishable.

Formally, we extend behavioural equivalences to the graded case:
bisimilarity and trace equivalence of Feller-Dynkin processes are a particular case of bisimilarity and trace equivalence of graded monad coalgebras (\Cref{th:fd-bisimilarity,th:fd-trace-equivalence});
graded monads compose via graded distributive laws (\Cref{prop:composite-graded-monad}), which we introduce similarly to monad-comonad graded distributive laws~\cite{gaboardi2016combining} to attach an output behaviour to graded coalgebras.
When the graded monad is accessible and its base category is locally presentable, we can prove the existence of a terminal graded coalgebra (\Cref{thm:graded-coalgebras-complete}). As in the ungraded case, terminal graded coalgebras characterise behavioural equivalence.

\subparagraph{Characteristic logics.}
Finally, we provide characteristic logics for behavioural equivalence and trace equivalence of graded coalgebras (\Cref{sec:characteristic-logics}).
We identify conditions for the Hennessy-Milner property: when logical equivalence coincides with behavioural or trace equivalence (\Cref{cor:branching-invariance,thm:branching-expressivity,cor:trace-invariance,thm:trace-expressivity}).

\subsection{Related work}

Universal coalgebra~\cite{rutten2000universal} models state-based systems as coalgebras of functors; coalgebra homomorphisms characterize behavioural equivalence. Trace equivalence in the coalgebraic setting has been modelled in different ways, including via distributive laws of monads over functors~\cite{jacobs2004trace,hasuo2007trace,sokolova2005phd}, via distributive laws of functors over monads~\cite{jacobs2015trace,goncharov2013trace}, or via corecursive algebras~\cite{rot2021steps}.
Coalgebraic methods also capture infinite traces~\cite{cirstea2010generic,cirstea2017branching,urabe2018coalgebraic}, and 
trace equivalence of coalgebras in a category of presheaves~\cite{hoshino2014memoryful,muroya2016memoryful,di_lavore2022monoidal}.
Coalgebraic logics capture different notions of equivalences on the state space~\cite{klin2007coalgebraic,schroder2008expressivity,cirstea2009modal,schroder_et_al:LIPIcs.STACS.2009.1855,kupke2011coalgebraic,gorin2013simulations,klin2017coalgebraic}.
Graded semantics models different types of process equivalence via monads graded by natural numbers~\cite{milius_et_al:LIPIcs.CALCO.2015.253,DorschMS19,ford2022graded}.
Graded characteristic logics give the logic counterpart of graded semantics~\cite{DBLP:conf/lics/ForsterSW25,ForsterSWBGM24}.
Continuous stochastic logic~\cite{BaierHHK00, DesharnaisP03} is a characteristic logic for continuous-time Markov processes where, contrary to our setting, individual jumps are considered observable in these logics.
(Stochastic) differential dynamic logics are interpreted over (stochastic) hybrid programs, which can evolve in continuous time~\cite{Platzer11,Platzer12}.

Graded monads are ubiquitous in program semantics~\cite{gaboardi2016combining,fujii2016formal,mcdermott2022flexibly}: they allow named nondeterministic choices~\cite{liellcock2025compositional,sarkis_et_al:LIPIcs.CALCO.2025.5}, control access to the source of randomness~\cite{lew2019trace} and control the dimension of quantum resources~\cite{abramsky_et_al:LIPIcs.MFCS.2017.35}.
Monadic graded effects can be combined with comonadic ones via graded distributive laws~\cite{gaboardi2016combining}.

The graded coalgebra conditions that we propose have appeared in different guises in other models of time-indexed transition systems:
transition systems whose evolution depends continuously on time may be expressed as arrows in the Kleisli category of a monad of continous paths on topological spaces~\cite{neves2016continuity}, where the conditions for a valid path are analogous to our graded coalgebra axioms;
another abstraction for timed processes considers partial actions of the monoid of time on the state space~\cite{kick2006coalgebraic} or, more generally, actions of the monoid of time in a monoidal category~\cite{lawvere1984functorial}---the action axioms correspond to the graded coalgebra axioms.
Other coalgebraic approaches to hybrid systems focus on their discrete component and treat the continuous component as observations of a system with discrete transitions~\cite{neves2018languages}.
Hybrid computations may also be modelled as morphisms in the Kleisli category of an appropriate Elgot monad~\cite{goncharov2019adequate,neves2025adequate}.

Our examples take inspiration from existing categorical abstractions of continuous-time systems, like those for Feller-Dynkin processes~\cite{chen2023behavioural}, Lawvere dynamical systems~\cite{lawvere1984functorial} and coalgebraic probabilistic systems with both discrete and continuous state spaces~\cite{sokolova2011survey}.
Final coalgebras in measurable spaces have also been studied before~\cite{kerstan2012coalgebraic,kerstan2013coalgebraic}.

\section{Graded coalgebras of monads}%
\label{sec:graded-coalgebras}
We start by recalling graded monads. In particular, the case of a monad graded by a monoidal category~\cite{benabou1967bicategories,smirnov2008,fujii2016formal} that arises when a monoid is regarded as a discrete monoidal category.

\begin{define}[Graded monad]%
  \label{def:graded-monad}
  A \intro[graded monad]{monad \(M\) graded by a monoid} \((T, \cdot, e)\) consists of a family of monoid-indexed endofunctors, \(M_{t} \colon \cat{C} \to
  \cat{C}\) for \(t \in T\), equipped with a \emph{graded multiplication}—a
  family of natural transformations \(\mu^{s,t} \colon M_{s}M_{t} \to M_{s \cdot
  t}\) for \(s,t \in T\)—and a \emph{graded unit}—a natural transformation
\(\eta \colon \id \to M_{e}\)—satisfying graded associativity, \(M_r(\mu^{s,t}) \dcomp \mu^{r, s \cdot t} = \mu^{r,s}_{M_t} \dcomp \mu^{r \cdot s, t}\), and graded unitality, \(M_t(\eta) \dcomp \mu^{t,e} = \id = \eta_{M_t} \dcomp \mu^{e,t}\) (also in \Cref{eq:graded-monad-diagrams}).
\end{define}
\begin{toappendix}
The diagrams below express associativity and unitality of graded monads (\Cref{def:graded-monad}).
\begin{equation}%
  \label{eq:graded-monad-diagrams}
  \begin{tikzcd}
    {M_r M_s M_t} \arrow{r}{M_r \mu^{s,t}} \arrow{d}[swap]{\mu^{r,s} M_t} & {M_r M_{s \cdot t}} \arrow{d}{\mu^{r, s \cdot t}} & {M_t} \arrow{r}{\eta M_t} \arrow{d}[swap]{M_t \eta} \arrow{dr}{\id} & {M_e M_t} \arrow{d}{\mu^{e,t}} \\
    {M_{r \cdot s} M_t} \arrow{r}[swap]{\mu^{r \cdot s, t}} & {M_{r \cdot s \cdot t}} & {M_t M_e} \arrow{r}[swap]{\mu^{t,e}} & {M_{t}}
  \end{tikzcd}
\end{equation}
\end{toappendix}

\noindent We now introduce the central concept of the present work.

\begin{define}[Graded coalgebra]%
  \label{def:gradedcoalgebra}
  \AP A \intro{graded coalgebra} of a \(T\)-\kl{graded monad} \(M_t \colon \cat{C}
  \to \cat{C}\) is a carrier object, \(X \in \cat{C}\), with a family of morphisms
  \(\gamma_{t} \colon X \to M_{t} X\) indexed by the monoid \((T, \cdot , e)\)
  and making the following diagrams commute.
  \[\begin{tikzcd}
    {X} \arrow{d}[swap]{\gamma_s} \arrow{r}{\gamma_{s \cdot t}} & {M_{s \cdot t} X} & {X} \arrow{r}{\gamma_e} \arrow{d}[swap]{\eta} & {M_e X} \\
    {M_s X} \arrow{r}[swap]{M_s \gamma_t} & {M_s M_t X} \arrow{u}[swap]{\mu^{s,t}} & {M_e X} \arrow[dash]{ur}[swap]{\id} & { }
  \end{tikzcd}\]
\end{define}

\noindent Graded coalgebras can be interpreted as transition systems whose clock ticks in $T$, a monoid representing time. The graded coalgebra axioms have a natural interpretation in these terms: unitality imposes that a transition of no time should produce no observable behaviour; multiplicativity imposes that a transition of \(s \cdot t\) time should produce the same behaviour as an \(s\)-transition followed by a \(t\)-transition.

\begin{remark}
  In the ungraded case, monad coalgebras are degenerate: the coalgebra must coincide with the monad unit. Grading, instead, allows nontrivial monad coalgebras.
\end{remark}

\begin{toappendix}
\begin{example}%
  \label{ex:solution-rep-system}
  The transition kernel of the Markov chain in~\eqref{eq:transition-graph-rep-system} is computed by solving an ordinary linear differential equation \(\vec{x}'(t) = A\vec{x}(t)\) determined by the chain, where the matrix \(A\) is given below.
  \[A =
    \begin{pmatrix}
      -2\mu & \lambda & \lambda & 0\\
      \mu & -(\lambda + \mu) & 0 & \lambda \\
      \mu & 0 & -(\lambda + \mu) & \lambda \\
      0 & \mu & \mu & -2\lambda
    \end{pmatrix}
  \]
  The solutions to ordinary linear differential equations have a particular shape: it is a sum of exponentials, \(\vec{x}(t) = \sum_{i = 1}^{n} c_{i} \vec{v}_{i} e^{l_{i}t}\), where \(l_{i}\) are the eigenvalues of \(A\) and \(v_{i}\) the corresponding eigenvectors.
  The constants \(c_{i}\) are found by imposing the initial conditions.
  For our example, the matrix of eigenvectors, \(V\), and the vector of eigenvalues, \(\vec{l}\), are given below.
  \[V =
    \begin{pmatrix}
      \lambda^{2} & 0 & -\lambda & 1 \\
      \lambda \mu & -1 & \lambda -\mu & -1 \\
      \lambda \mu & 1 & 0 & -1 \\
      \mu^{2} & 0 & \mu & 1
    \end{pmatrix}
    \qquad
    \vec{l} =
    \begin{pmatrix}
      0\\
      -(\lambda + \mu)\\
      -(\lambda + \mu)\\
      -2(\lambda + \mu)
    \end{pmatrix}
  \]
  For each state \(j\), we impose the base vector \(\vec{e}_{j}\) as initial condition to find constants \(c_{i,j}\).
  With these constants, we obtain the transition kernel \(\gamma_{t}(x_{k} \mid x_{j}) = \sum_{i = 1}^{4} c_{i,j} v_{k,i} e^{l_{i}t}\).
  In our case, the matrix \(C\) of the constants \(c_{i,j}\) is below.
  \[C = \frac{1}{(\lambda + \mu)^{2}}
    \begin{pmatrix}
      1 & 1 & 1 & 1\\
      \mu(\mu - \lambda) & -2\lambda \mu & \lambda^{2}+\mu^{2} & \lambda (\lambda - \mu)\\
      -2\mu & \lambda - \mu & \lambda - \mu & 2 \lambda \\
      \mu^{2} & -\lambda \mu & -\lambda \mu & \lambda^{2}
    \end{pmatrix}
  \]
  For general reasons, this transition kernel satisfies the axioms of a graded monad coalgebra; in this example, it is not difficult to check these properties by hand: by construction, \(\gamma_{0} = \id\); by the relationship between \(c_{i,j}\) and \(v_{k,j}\), \(\gamma_{s} \dcomp \gamma_{t} = \gamma_{s+t}\).
\end{example}
\end{toappendix}

\begin{example}%
  \label{ex:graded-coalgebra-rep-system}
  The Markov chain~\eqref{eq:transition-graph-rep-system} in the previous section determines a graded coalgebra \(\gamma_{t} \colon 4 \to \distr(4)\) that encodes the solution to the system of differential equations specified by the graph in~\eqref{eq:transition-graph-rep-system}.
  Its expression is a sum of exponentials, \(\gamma_{t}(k \mid j) = \sum_{i = 1}^{4} c_{i,j} v_{k,i} e^{l_{i}t}\); the values of the constants are not relevant for our discussion but can be found in \Cref{ex:solution-rep-system}.
  From this general shape, and using the conditions that determine the constants \(c_{i,j}\) and \(v_{j,k}\), one can check that the kernels \(\gamma_{t}\) satisfy the axioms of a graded coalgebra.
  We highlight that graded coalgebras are semantic models of transition systems, in the sense that they encode their trajectories.
  On the other hand, continuous-time transition systems are often specified more syntactically, like the Markov chain in~\eqref{eq:transition-graph-rep-system}, encoding the one-step transition behaviour of the system.
  We leave as future work to categorically express the relationship between these two encodings.
\end{example}

\subsection{Behavioural and trace equivalences of graded coalgebras}

Intepreting coalgebras as state-based systems raises the question of how to
characterize behavioural equivalence: whether an outside observer can
distinguish between two given states. Following the coalgebra literature, we
define behavioural equivalence via cospans of coalgebra homomorphisms.
\begin{define}%
  \label{def:graded-coalgebra-morphism}
  A homomorphism, \(h \colon (X, \gamma) \to (Y, \delta)\), between two graded coalgebras of a \(T\)-graded monad \(M_{t} \colon \cat{C} \to \cat{C}\) is a morphism \(h\colon X \to Y\) in \(\cat{C}\) commuting with the coalgebras, \(\gamma_t \dcomp M_t(h) = h \dcomp \delta_t\) (\Cref{eq:graded-coalgebra-morphism}), for all \(t\in T\).
  Graded coalgebra homomorphisms between graded coalgebras for a graded monad
  \(M\) form a category, \(\GCoAlg{M}\).
\end{define}
\begin{toappendix}
The diagram below defines graded coalgebra morphisms (\Cref{def:graded-coalgebra-morphism}).
\begin{equation}%
  \label{eq:graded-coalgebra-morphism}
  \begin{tikzcd}
    X \arrow{r}{\gamma_t} \arrow{d}[swap]{h} & M_tX \arrow{d}{M_th} \\
    Y \arrow{r}[swap]{\delta_t} & M_tY
  \end{tikzcd}
\end{equation}
\end{toappendix}

\begin{define}[Behavioural equivalence]%
  \label{def:behavioural-equivalence}
  Let \(M\) be a \(T\)-graded monad on a cartesian category \(\cat{C}\) and let \((X, \gamma)\) and \((Y, \delta)\) be two graded \(M\)-coalgebras.
  Two states \(x \colon 1 \to X\) and \(y \colon 1 \to Y\) are \emph{behaviourally equivalent} if there is some graded \(M\)-coalgebra \((Z, \xi)\) and two coalgebra homomorphisms \(g \colon (X, \gamma) \to (Z, \xi)\) and \(h \colon (Y, \delta) \to (Z, \xi)\) such that \(x \dcomp g =
  y \dcomp h\).
\end{define}

\begin{toappendix}
\begin{example}%
  \label{ex:solution-bisimilar-rep-system}
  The transition kernel of the Markov chain in~\eqref{eq:bisimilar-rep-system} can be found in the same way as the previous one (\Cref{ex:solution-rep-system}).
  For this second Markov chain, the matrix of the system of ordinary differential equations is \(B\) as given below.
  \[B =
    \begin{pmatrix}
      -2 \mu & \lambda & 0\\
      2 \mu & -(\lambda + \mu) & 2 \lambda \\
      0 & \mu & -2\lambda
    \end{pmatrix}
  \]
  Its matrix of eigenvectors \(W\) and vector of eigenvalues \(\vec{m}\) are below.
  \[W =
    \begin{pmatrix}
      \lambda^{2} & -\lambda & 1 \\
      2 \lambda \mu & \lambda - \mu & -2\\
      \mu^{2} & \mu & 1
    \end{pmatrix}
    \qquad
    \vec{m} =
    \begin{pmatrix}
      0\\
      -(\lambda + \mu)\\
      -2(\lambda + \mu)
    \end{pmatrix}
  \]
  The matrix \(D\) of constants can be found as in \Cref{ex:solution-rep-system}.
  \[D = \frac{1}{(\lambda + \mu)^{2}}
    \begin{pmatrix}
      1 & 1 & 1 \\
      -2\mu & \lambda - \mu & 2 \lambda \\
      \mu^{2} & -\lambda \mu & \lambda^{2}
    \end{pmatrix}
  \]
  We prove that the function \(h\) from \Cref{ex:bisim-rep-system} is a graded coalgebra homomorphism.
  We can express the transition kernels as matrix multiplications, \(\gamma_{t} = V \cdot \overline{C}_{t}\) and \(\delta_{t} = W \cdot \overline{D}_{t}\), where \(\overline{C}_{t}(i,j) = c_{i,j} e^{l_{i}t}\) and \(\overline{D}_{t}(i,j) = d_{i,j} e^{m_{i}t}\).
  Then, the entry \((k,j)\) of the matrix \(\gamma_{t}\) is the transition probability \(\gamma_{t}(k \mid j)\) from state \(j\) to state \(k\).
  The function \(h\), written in matrix form, is below left.
  \[H =
    \begin{pmatrix}
      1 & 0 & 0 & 0\\
      0 & 1 & 1 & 0\\
      0 & 0 & 0 & 1
    \end{pmatrix}
    \qquad
    Q =
    \begin{pmatrix}
      1 & 0 & 0 & 0\\
      0 & 0 & 0 & 0\\
      0 & 0 & 1 & 0\\
      0 & 0 & 0 & 1
    \end{pmatrix}
  \]
  We check that \(H \cdot \gamma_{t} = \delta_{t} \cdot H\), where we use the auxiliary matrix \(Q\) defined above right.
  \begin{align*}
    H \cdot \gamma_{t} = H \cdot V \cdot \overline{C}_{t} = W \cdot H \cdot Q \cdot \overline{C}_{t} = W \cdot \overline{D}_{t} \cdot H = \delta_{t} \cdot H
  \end{align*}
\end{example}
\end{toappendix}

\begin{example}%
  \label{ex:bisim-rep-system}
  There is a morphism of graded coalgebras that witnesses behavioural equivalence between states in the Markov chains~(\ref{eq:transition-graph-rep-system} and~\ref{eq:bisimilar-rep-system}) from the previous section.
  We indicate with \(\gamma_{t}\) the transition kernel of the chain~\eqref{eq:transition-graph-rep-system} (\Cref{ex:graded-coalgebra-rep-system}), and with \(\delta_{t}\) the transition kernel of the chain~\eqref{eq:bisimilar-rep-system} (see \Cref{ex:solution-bisimilar-rep-system}).
  The morphism \(h \colon \{0,L,R,2\} \to \{0,1,2\}\), defined as \(h(0) = 0\), \(h(2) = 2\) and \(h(L) = h(R) = 1\), is a graded coalgebra morphism \(h \colon (\{0,L,R,2\}, \gamma_{t}) \to (\{0,1,2\}, \delta_{t})\).
  See \Cref{ex:solution-bisimilar-rep-system} for the details.
\end{example}

\noindent For trace equivalence we follow work on graded semantics~\cite{DorschMS19}: the trace of length \(t \in T\) of a state \(x \colon 1 \to X\) is obtained by running the system for \(t\) time, starting in \(x\), and discarding the state space.

\begin{define}[Trace equivalence]%
  \label{def:trace-equivalence}
  Let \(M\) be a \(T\)-graded monad on a cartesian category \(\cat{C}\) and let \((X, \gamma)\) and \((Y, \delta)\) be two graded \(M\)-coalgebras.
  Two states \(x \colon 1 \to X\) and \(y \colon 1 \to Y\) are \emph{trace equivalent} if \(x \dcomp \gamma_{t} \dcomp M_{t}(!_X) = y \dcomp \delta_{t} \dcomp M_{t}(!_Y)\), for all grades \(t \in T\), where \(!_X \colon X \to 1\) is the unique morphism to the terminal object.
\end{define}

\noindent Graded coalgebras are a common generalisation of functor coalgebras, Lawvere dynamical systems and graded semantics, as we show next.
\Cref{sec:fd-processes} will use graded coalgebras to express Feller-Dynkin processes, and capture their behavioural and trace equivalence.

\subsection{First examples}%
\label{ex:graded-coalgebras}
\subparagraph{Coalgebras of functors.}%
\label{exp:functor-coalgebra}

Intuitively, coalgebras of functors are an abstraction of discrete-time transition systems~\cite{rutten2000universal}:
a morphism \(\gamma \colon X \to FX\) expresses a one-step transition, with the functor \(F\) regulating its possible behaviour.
Formally, any endofunctor $F$ induces a \kl[graded monad]{graded monad} $F^{n}$ given by its $n$-fold composition: its multiplications and unit are the identities \(F^{m}F^{n} = F^{m+n}\) and \(\id = F^{0}\).
In this sense, functor coalgebras are \((\naturals, +, 0)\)-graded coalgebras.
A minimalistic concrete example is the random walk.

\begin{propositionrep}%
  \label{prop:functor-coalgebras}
  \AP The category of $F$-coalgebras is isomorphic to the category of \((\naturals, +, 0)\)-graded coalgebras for the graded monad \(F^{n}\).
\end{propositionrep}
\begin{proof}[Proof sketch]
  A coalgebra \(\gamma_{1} \colon X \to FX\) of a functor \(F\) determines morphisms \(\gamma_{n+1} = \gamma_{n} \dcomp F^{n}(\gamma_{1})\)~\cite{milius_et_al:LIPIcs.CALCO.2015.253,DorschMS19}; these satisfy the graded coalgebra axioms, with \(\gamma_{0} = \id_{X}\).
  Viceversa, a graded coalgebra \(\gamma_{n} \colon X \to F^{n}X\) includes a coalgebra \(\gamma_{1} \colon X \to FX\) of the functor \(F\).
  These transformations extend to coalgebra morphisms as the identity on the underlying \(\cat{C}\)-morphisms.
  Thus, we obtain two functors that are each other's inverses and give an isomorphism between the category of coalgebras of the functor \(F\) and the category of \((\naturals, +, 0)\)-graded coalgebras for the graded monad \(F^{n}\).
\end{proof}

\begin{example}[Random walk]%
  \label{ex:random-walk}
  Consider the \intro{finitely-supported distribution monad}, \(\distr \colon \Set \to
  \Set\), and define \(\gamma_{1} \colon \integers \to \distr(\integers)\) as
  \(\gamma_{1}(x) = 0.5 \ket{x-1} + 0.5 \ket{x+1}\):
  from the state \(x\), the system transitions one step to the left
  or one step to the right with equal probability. This is a coalgebra for the
  functor underlying the distribution monad \(\distr\), which we have seen
  coincide with \((\naturals, +, 0)\)-graded coalgebras for the graded monad
  \(\distr^{n}\).
\end{example}

\subparagraph{Lawvere dynamical systems.}
Any monad (e.g., \(\distr\) from \Cref{ex:random-walk}) can be reinterpreted as a constantly-graded monad:
\(\distr_{n} = \distr\) for all \(n \in \naturals\).
This helps us recover Lawvere dynamical systems~\cite{lawvere1984functorial}, which consist of endomorphisms \(\gamma_{t}\) indexed by elements \(t\) of a monoid \((T, \cdot, e)\) that are compatible with the monoid structure\footnote{We consider the external version of Lawvere dynamical systems. The internal version employs monoid actions \(\gamma \colon T \tensor X \to X\) in some monoidal category.}.

\begin{define}%
  \AP A \intro{Lawvere dynamical system} in a category \(\cat{D}\) is a monoid
  homomorphism \(\gamma \colon (T, \cdot, e) \to (\cat{D}(X,X), (\dcomp), \id)\)
  to the monoid of endomorphisms on some object \(X\).
  Explicitly, a \kl{Lawvere dynamical system} is a family of morphisms \(\gamma_{t} \colon X \to X\) in \(\cat{D}\) that preserve the unit, \(\gamma_{e} = \id_{X}\), and the multiplication, \(\gamma_{s \cdot t} = \gamma_{s} \dcomp
  \gamma_{t}\).
\end{define}

\begin{toappendix}
More generally, monoid morphisms regrade monads.

\begin{lemma}%
  \label{lemma:regrading}
  Monoid morphisms induce regrading of monads: given a \(U\)-graded monad \((M, \mu^{u,v}, \eta)\) and a monoid morphism \(h \colon T \to U\), the family of functors \(M_{h(t)}\) assemble into a \(T\)-graded monad \(M_{h}\) with graded multiplication \(\mu^{s,t} = \mu^{h(s),h(t)}\) and unit \(\eta\).
\end{lemma}
\begin{proof}
  The multiplication \(\mu^{h(s),h(t)}\) has the correct type because \(h(s) \cdot h(t) = h(s \cdot t)\);
  similarly, the unit \(\eta\) has the correct type because \(h(e) = e\).
  The graded associativity and unitality equations hold because they hold for \(M\).
\end{proof}

\noindent Monads are \(1\)-graded monads and, for any monoid \(T\), there is a unique morphism to the terminal monoid \(1\).
By \Cref{lemma:regrading}, we may regrade any monad \(M\) with this morphism and obtain the constantly graded monad \(M_{t} = M\).
\end{toappendix}
\begin{propositionrep}
  \kl{Lawvere dynamical systems} in the Kleisli category of a monad \(M \colon \cat{C} \to \cat{C}\) are graded coalgebras for the constantly-graded monad associated to \(M\).
\end{propositionrep}
\begin{proof}[Proof sketch]
  A \kl{Lawvere dynamical system} in \(\kleisli{M}\)
  is a family of morphisms \(\gamma_{t} \colon X \to MX\) in \(\cat{C}\).
  The monoid morphism axioms in \(\kleisli{M}\) are exactly the graded coalgebra axioms for the constantly-graded monad \(M_{t} = M\).
\end{proof}

\begin{example}[Markov monoids]
  Markov semigroups~\cite[Definition~14.40]{klenke2008probability} are Lawvere dynamical systems for the \intro{Giry monad}, \(\giry \colon \Meas \to \Meas\), on measurable spaces~\cite{giry1982categorical}; time is often restricted to the nonnegative reals, \((\posreals, +, 0)\).
  The condition of compatibility with the unit, \(\alpha_{e} = \id\), suggests \emph{Markov monoids} as a more appropriate name for Markov semigroups.
  Brownian motion is an example of Markov monoid (\Cref{ex:brownian}).
  Analogously, we could consider \emph{partial Markov monoids} to be Lawvere dynamical systems for the \intro{subprobability measure monad}, \(\subgiry = \giry(- + 1)  \colon \Meas \to \Meas\)~\cite{panangaden1999category}.
\end{example}

\begin{example}[Timed transition systems]
  Timed transition systems~\cite{kick2006coalgebraic,kick2003thesis} are Lawvere dynamical systems in the category of sets and partial functions (cf.~\cite[Proposition~2.11]{kick2006coalgebraic}) with additional conditions on the monoid of time.
\end{example}

\begin{toappendix}
\begin{example}[Brownian motion]%
  \label{ex:brownian}
  The family of morphims $β_s ፡ \reals → \giry(\reals)$ defining Brownian motion, $β_s(x) = \operatorname{Normal}(x;s)$, form a Markov monoid, i.e.\ a
  \((\posreals, + , 0)\)-graded coalgebra for the constantly graded Giry monad.
  The coalgebra axioms impose that $y \sim \operatorname{Normal}(x;s)$ and $z \sim \operatorname{Normal}(y;t)$ imply that $z \sim \operatorname{Normal}(x;s+t)$;
  and that $y \sim \operatorname{Normal}(x;0)$ implies $y = x$.
\end{example}
\end{toappendix}

\subparagraph{Graded semantics.}%
\label{sec:graded-semantics}

Graded semantics~\cite{DorschMS19} captures different kinds of state-based process equivalence.
Systems are modelled as coalgebras \(\gamma \colon X \to GX\) of a \(\Set\)-endofunctor \(G\); the notion of process equivalence is determined by a \((\mathbb{N}, +, 0)\)-graded monad \(M\).

\begin{define}[{\cite[Definition~5.1]{milius_et_al:LIPIcs.CALCO.2015.253}}]
  A \intro{graded semantics} for a functor \(G \colon \Set \to \Set\) in a \((\mathbb{N}, +, 0)\)-graded monad \(M\) is a natural transformation \(\alpha \colon G \to M_{1}\).
\end{define}

\begin{toappendix}
We show that a graded semantics for a functor-coalgebra defines a \((\mathbb{N}, +, 0)\)-graded coalgebra.
For this, we need the notion of morphism of graded monads.

\begin{define}%
  \label{def:graded-monad-morphism}
  \AP A \intro{graded monad morphism} \(\alpha \colon M \to N\) between \((T, \cdot , e)\)-graded monads \(M,N \colon \cat{C} \to \cat{C}\) is a family of natural transformations \(\alpha^{t} \colon M_{t} \to N_{t}\) indexed by the monoid \((T,\cdot, e)\) that commute with the graded multiplication and unit, \((\alpha^{s} \pll \alpha^{t}) \dcomp \mu^{s,t} = \mu^{s,t} \dcomp \alpha^{s \cdot t}\) and \(\eta \dcomp \alpha^e = \eta\) (\Cref{eq:graded-monad-morphism}), where \intro[pll]{}\((\pll)\) indicates the parallel composition of natural transformations.
\begin{equation}%
  \label{eq:graded-monad-morphism}
  \begin{tikzcd}
    {M_{s} M_{t}} \arrow{r}{\alpha^{s} \pll \alpha^{t}} \arrow{d}[swap]{\mu^{s,t}} & {N_{s} N_{t}} \arrow{d}{\mu^{s,t}} & {\id_{\cat{C}}} \arrow{dr}{\eta} \arrow{d}[swap]{\eta} & {}\\
    {M_{s \cdot t}} \arrow{r}[swap]{\alpha^{s \cdot t}} & {N_{s \cdot t}} & {M_{e}} \arrow{r}[swap]{\alpha^{e}} & {N_{e}}
  \end{tikzcd}
\end{equation}
\end{define}

\begin{lemma}%
  \label{lemma:graded-morphism}
  For a functor \(G \colon \cat{C} \to \cat{C}\) and a \((\naturals,+,0)\)-graded monad \(M\) on \(\cat{C}\), a natural transformation \(\alpha \colon G \to M_{1}\) induces a morphism of \((\mathbb{N}, +, 0)\)-graded monads \(\alpha^{n} \colon G^{n} \to M_{n}\).
  The morphism is defined by induction: \(\alpha^{0} = \eta\) and \(\alpha^{n+1} = (\alpha \pll \alpha^{n}) \dcomp \mu^{1,n}\).
\end{lemma}

\end{toappendix}
\noindent
A coalgebra with a graded semantics determines a graded coalgebra.

\begin{propositionrep}%
  \label{prop:graded-semantics}
  A \kl{graded semantics} \((M_{n}, \alpha)\) for a coalgebra \(\gamma \colon X \to GX\) extends uniquely to a graded coalgebra \(\gamma_{n} \colon X \to G^{n}X\) for the \((\naturals,+,0)\)-graded monad \(G^{n}\) and a \kl{graded monad morphism} \(\alpha^{n} \colon G^{n} \to M_{n}\).
\end{propositionrep}
\begin{proof}[Proof sketch]
  \Cref{prop:functor-coalgebras} gives a graded \(G^{n}\)-coalgebra \(\gamma_{n}\), and \Cref{lemma:graded-morphism} gives a graded monad morphism \(\alpha^{n} \colon G^{n} \to M_{n}\).
  The composition \(\gamma_{n} \dcomp \alpha^{n}_{X}\), then, gives a graded coalgebra of \(M_{n}\).
\end{proof}

\noindent By \Cref{prop:graded-semantics}, we obtain a graded coalgebra \(\delta_{n} = \gamma_{n} \dcomp \alpha^{n}\) for the graded monad \(M\) by composing the graded \(G\)-coalgebra with the graded monad morphism.
Indeed, process equivalence in graded semantics inspired our definition of trace equivalences.

\section{Feller-Dynkin processes via graded distributive laws}%
\label{sec:fd-processes}%

Behavioural equivalence collapses in the absence of an explicit output.
For instance, given a coalgebra of the functor $\distr$, all its states are behaviourally equivalent.
In the ungraded case, this problem can be solved by attaching observations to the functor itself: e.g., $B \times (\distr{-})$ allows for probabilistic transitions, while each state produces an observation in $B$. In the graded case, the situation is subtler, as one needs to retain the graded monad structure while incorporating the observation. 

To this end, we introduce labelled graded coalgebras and propose them as a generalisation of Feller-Dynkin processes where a graded monad determines the branching behaviour, and a functor determines the observation behaviour.
We prove that labelled graded coalgebras are instances of graded coalgebras when the branching effect and the labelling effect can be combined via a graded distributive law.

\subsection{Labelled graded coalgebras}

The behaviour of a transition system is often composed of two parts: the branching behaviour, regulated by a monad, \(M\), and the observation behaviour, regulated by a functor, \(F\).
The observation functor may appear inside the monad, \(MF\), or outside it, \(FM\), giving rise to Kleisli and Eilenberg-Moore semantics, respectively~\cite{hasuo2007trace,jacobs2015trace}.
We bring this distinction to the graded setting.

\begin{define}%
  \AP A \intro{Kleisli-labelled coalgebra} (resp. an  \intro{Eilenberg-Moore-labelled coalgebra}) of a $T$-graded monad $M_t$ on a category \(\cat{C}\) labelled by an endofunctor $F \colon \cat{C} \to \cat{C}$ is a carrier object $X \in \cat{C}$ together with a \(T\)-graded \(M\)-coalgebra $\gamma_t \colon X \to M_tX$ and a morphism \(l \colon X \to M_{e}FX\) (resp. \(l \colon X \to FM_{e}X\)).
\end{define}

\noindent The graded coalgebra \(\gamma_{t}\) gives the transitions; the morphism \(l\) labels the internal states with some observation.

\begin{toappendix}
\begin{example}%
  \label{ex:streams-labelled-coalgebra}
  Consider the labelling functor \(F =(B \times -)\) and the \(\posreals\)-graded monad that is the constantly-graded identity functor on \(\Set\).
  In this case, Kleisli-labelled and Eilenberg-Moore-labelled coalgebras coincide and are functions \(\gamma_{t} \colon X \to X\), such that \(\gamma_{0} = \id\) and \(\gamma_{s} \dcomp \gamma_{t} = \gamma_{s+t}\), together with a labelling function \(l \colon X \to B \times X\).
\end{example}
\end{toappendix}

\begin{example}%
  \label{ex:finitary-fd}%
  \label{ex:fd-process-rep-system}%
  \label{ex:fd-processes-labelled-coalgebra}
  A \intro{finitary Feller-Dynkin process} consists of a finitary partial Markov monoid, $\gamma_t ፡ X → \subdistr(X)$, and an observation function, $\obs \colon X \to B$, where \(\subdistr = \distr(- +1)\) is the \intro{finitely-supported subdistribution monad}.
  As a concrete example, the graded coalgebra in \Cref{ex:graded-coalgebra-rep-system} together with the observation function \(o \colon \{0,L,R,2\} \to \distr(\{yes, no\})\) in~\Cref{sec:intro}, is a finitary Feller-Dynkin process with observations in \(\distr(\{yes,no\})\).

  Finitary Feller-Dynkin processes determine Klei\-sli-la\-belled and Eilenberg-Moore-labelled coalgebras for the labelling functor \((B \times -)\):
  given a finitary Feller-Dynkin process \((\gamma,\obs)\), we take \(\gamma\) as the transition structure;
  the labelling morphism for the Kleisli case is \(l = \langle \obs, \id_{X} \rangle \dcomp \eta_{B \times X}\), while the labelling morphism for the Eilenberg-Moore case is \(l = \langle \obs, \eta_{X} \rangle\).
\end{example}

\begin{example}
  A \intro{Feller-Dynkin process}~\cite[Definition~13]{chen2023behavioural} is a family of transition kernels \(\gamma_{t} \colon X \to \subgiry(X)\) together with a measurable function \(\obs \colon X \to 2^{A}\), for some finite set \(A\), the \emph{atomic propositions}, considered with the discrete \(\sigma\)-algebra; the transition kernels need to satisfy time-compatibility, \(\gamma_{s} \dcomp \gamma_{t} = \gamma_{s+t}\) and \(\gamma_{0} = \id_{X}\) in \(\kleisli{\subgiry}\), and some continuity conditions, which we disregard.
  Intuitively, the observation function maps each state to the set of atomic propositions that hold in that state.

  As in the finitely supported case, a \kl{Feller-Dynkin process} determines both a Kleisli-labelled and an Eilenberg-Moore-labelled coalgebra: the transition structure are the morphisms \(\gamma_{t}\); the labelling morphism for the Kleisli case is \(l = \langle \obs, \id_{X} \rangle \dcomp \eta_{2^{A} \times X}\), while for the Eilenberg-Moore case is \(l = \langle \obs, \eta_{X} \rangle\).
  However, contrary to the finitely supported case, \kl{Feller-Dynkin processes} are a proper subclass of labelled graded coalgebras: \kl{Feller-Dynkin processes} additionally require that trajectories be \emph{cadlag} (i.e.\ right-continuous with left limit).
  Graded coalgebras of monads, as presently defined, cannot natively impose continuity conditions with respect to time; these continuity conditions may be added by enriching in topological spaces, but we leave this direction as future work.
\end{example}

\noindent We now go on to show that labelled graded coalgebras are particular instances of graded coalgebras when we can combine the branching effect of the graded monad with the effect of the labelling functor.
This is done via graded distributive laws, which we introduce next.

\subsection{Graded distributive laws}

A distributive law \(\lambda \colon PM \to MP\) of a monad \(P\) over a monad \(M\) gives a composite monad \(MP\)~\cite{beck1969distributive}.
We extend the definition of distributive law to the graded case and prove that graded distributive laws give composite graded monads.
The definition of graded distributive law between monads is analogous to that of monad-comonad graded distributive law~\cite{gaboardi2016combining}.

\begin{define}[Graded distributive law]
  A \intro{graded distributive law} of a \(U\)-graded monad \(P_{u} \colon \cat{C} \to \cat{C}\) over a \(T\)-graded monad \(M_{t} \colon \cat{C} \to \cat{C}\) is a family of natural transformations \(\lambda^{u,t} \colon P_{u}M_{t} \to M_{t}P_{u}\) indexed by the product monoid \(T \times U\) that commutes with graded multiplications and units as below.\\
  \begin{tikzcd}[column sep = normal, scale cd=0.85]
    {P_u M_s M_t} \arrow{d}[swap]{\lambda^{u,s}_{M_t}} \arrow{r}{P_u \mu^{s,t}}
    & {P_u M_{s \cdot t}} \arrow{dd}{\lambda^{u,s \cdot t}}\\
    {M_s P_u M_t} \arrow{d}[swap]{M_s \lambda^{u,t}}
    & \\
    {M_s M_t P_u} \arrow{r}{\mu^{s,t}_{P_u}} 
    & {M_{s \cdot t} P_u}
  \end{tikzcd}\
  \begin{tikzcd}[column sep = normal, scale cd=0.85]
    {P_u P_v M_t} \arrow{d}[swap]{P_u \lambda^{v,t}} \arrow{r}{\mu^{u,v}_{M_t}}
    & {P_{u \cdot v} M_t} \arrow{dd}{\lambda^{u \cdot v, t}} \\
    {P_u M_t P_v} \arrow{d}[swap]{\lambda^{u,t}_{P_v}}
    & \\
    {M_t P_u P_v} \arrow{r}{M_t \mu^{u,v}}
    & {M_t P_{u \cdot v}}
  \end{tikzcd}\
  \begin{tikzcd}[column sep = normal, scale cd=0.85]
    {P_u} \arrow{r}{P_u \eta} \arrow{dr}[swap]{\eta_{P_u}} & {P_u M_e} \arrow{d}{\lambda^{u,e}} \\
    {} & {M_e P_u}
  \end{tikzcd}\
  \begin{tikzcd}[column sep = normal, scale cd=0.85]
    {M_t} \arrow{r}{\eta_{M_t}} \arrow{dr}[swap]{M_t \eta} & {P_e M_t} \arrow{d}{\lambda^{e,t}} \\
    {} & {M_t P_e}
  \end{tikzcd}
\end{define}

\begin{theoremrep}%
  \label{prop:composite-graded-monad}
  A graded distributive law \(\lambda^{u,t} \colon P_{u}M_{t} \to M_{t}P_{u}\) of a \(U\)-graded monad over a \(P\) \(T\)-graded monad \(M\) induces a composite \(T \times U\)-graded monad \(M_{t}P_{u}\).
\end{theoremrep}
\begin{proof}
  This proof is analogous to the ungraded case because we restrict to grading by a monoid.
  The monad unit is the parallel composition of the two units: \(\eta = \eta^M \pll \eta^P \colon \id \to M_e P_e\);
  the graded monad multiplication uses the distributive law: \(\mu = (M_s \pll \lambda^{u,t} \pll P_v) \dcomp (\mu^{(M) s,t} \pll \mu^{(P) u,v}) \colon M_s P_u M_t P_v \to M_{s \cdot t} P_{u \cdot v}\).
  These are natural transformations because they are compositions of natural transformations.
  Unitality follows from the fact that the distributive law commutes with the units and unitality of the two graded monads;
  associativity follows from the fact that the distributive law commutes with the multiplications and associativity of the two graded monads.
\end{proof}

\noindent With an appropriate distributive law, the labelling morphism of a labelled coalgebra can be incorporated in a new transition structure, giving a coalgebra graded by a monoid of \emph{sampling intervals}.
In this way, labelled graded coalgebras, and Feller-Dynkin processes as a consequence, are particular instances of graded coalgebras.
The grading by the monoid of sampling intervals determines the order and amount by which the system is progressed and/or observed.
The monoid of sampling intervals is the result of the interaction between a \(T\)-graded monad, \(M_{t}\), and the \(\naturals\)-graded monad, \(F^{n}\), generated by a functor \(F\), via graded distributive laws.

\begin{define}%
  [Monoid of sampling intervals]%
  \label{def:sampling-monoid}
  \AP%
  For a monoid \((T, \cdot, e)\), consider the \intro[Samp]{}coproduct monoid \(\Samp_{T} = (T, \cdot, e) + (\naturals, +, 0)\).
  The elements of this monoid are lists \((t_{0},k_{0}, \dots, t_{n},k_{n})\) alternating elements \(t_{i} \in T\), with \(t_{i} \neq e\) for \(i = 1, \dots, n\), and natural numbers \(k_{i} \in \naturals\), with \(k_{i} > 0\) for \(i = 0, \dots, n-1\).
  Intuitively, an element \(t_{i}\) indicates a transition of \(t_{i}\) time, while the number \(k_{i}\) indicates the number of observations to perform after the transition.
  The multiplication is concatenation of lists followed by simplification, as in \Cref{eq:mult-sampling-interval}.
  The unit is the list \((e,0)\).
\end{define}
\begin{toappendix}
  The multiplication in the monoid of sampling intervals (\Cref{def:sampling-monoid}) is concatenation of lists followed by simplification.
  \begin{align}
    & (s_{0}, j_{0}, \dots, s_{m}, j_{m}) \cdot (t_{0},k_{0}, \dots, t_{n},k_{n}) \nonumber\\
    &= \begin{cases}
        (s_{0}, j_{0}, \dots, s_{m}, j_{m},t_{0},k_{0}, \dots, t_{n},k_{n}) & \text{if } j_{m} > 0 \text{ and } t_{0} \neq e\\
        (s_{0}, j_{0}, \dots, s_{m} \cdot t_{0},k_{0}, \dots, t_{n},k_{n}) & \text{if } j_{m} = 0\\
        (s_{0}, j_{0}, \dots, s_{m}, j_{m}+k_{0}, \dots, t_{n},k_{n}) & \text{if } t_{0}= e
      \end{cases}\label{eq:mult-sampling-interval}
  \end{align}
\end{toappendix}

\noindent \Cref{prop:composite-graded-monad} gives a composite monad graded by the product monoid, but the monoid of sampling intervals is, instead, a coproduct.
Regrading the composite monad solves this issue.
Monoid morphisms regrade monads (\Cref{lemma:regrading}).
As a consequence, a span of monoid morphisms regrades composite monads.
There is always a span of monoid morphisms from a coproduct to its components, \(p_{1} \colon U + V \to U\) and \(p_{2} \colon U + V \to V\) defined below.
\begin{align*}
  p_{1}(u_{0}, v_{0}, \dots, u_{n}, v_{n}) &= u_{0} \cdots u_{n} &
  p_{2}(u_{0}, v_{0}, \dots, u_{n}, v_{n}) &= v_{0} \cdots v_{n}
\end{align*}
With this span, we always obtain a regrading by the coproduct monoid.

\begin{propositionrep}%
  \label{cor:coproduct-distr-law}
  For three monoids \(T\), \(U\) and \(V\) with monoid morphisms \(h \colon T \to U\) and \(k \colon T \to V\), consider a \(U\)-graded monad \(M\) and a \(V\)-graded monad \(P\).
  A graded distributive law \(\lambda^{v,u} \colon P_vM_u \to M_uP_v\) of \(P\) over \(M\) induces a composite \(T\)-graded monad \(M_{h(t)} P_{k(t)}\).

  In particular, the graded distributive law \(\lambda\) induces a composite \(U+V\)-graded monad \(M_{u_{0} \cdots u_{n}} P_{v_{0} \cdots v_{n}}\) by considering the monoid morphisms \(p_{1} \colon U+V \to U\) and \(p_{2} \colon U+V \to V\).
\end{propositionrep}
\begin{proof}
  \Cref{prop:composite-graded-monad} gives a composite monad, \(M_u P_v\), graded by \(U \times V\);
  the span of monoid morphisms corresponds to a morphism to the product, \(\langle h,k \rangle \colon T \to U \times V\); 
  by \Cref{lemma:regrading}, we obtain a \(T\)-graded monad, \(M_{h(t)} P_{k(t)}\).
\end{proof}

\noindent The monoid morphism from the monoid of sampling intervals to the monoid of time \(T\), \(l \colon \Samp_{T} \to (T, \cdot, e)\), may be interpreted as returning the length of the sampling interval; the monoid morphism to the natural numbers, \(c \colon \Samp_{T} \to (\naturals, +, 0)\), may be interpreted as returning the number of samples.

\begin{toappendix}
\subsection{Distributive laws between functors and graded monads}%
\label{sec:monad-functor-graded-distributive-laws}
\noindent Some \kl{graded distributive laws} arise from distributive laws between graded monads and functors.
As in the ungraded case, there are two possible combinations.
We elaborate on Kleisli-laws here, but the Eilenberg-Moore case is analogous (see \Cref{sec:graded-em-laws}).

\begin{define}
  A \intro{graded Kleisli-law} of a functor \(F \colon \cat{C} \to \cat{C}\) over a \(T\)-graded monad \(M_{t} \colon \cat{C} \to \cat{C}\) is a family of natural transformations \(\lambda^{t} \colon FM_{t} \to M_{t}F\) indexed by the monoid \(T\) that interchanges with the graded monad multiplication and unit as below.
  \begin{equation*}
    \begin{tikzcd}
      {F M_s M_t} \arrow{r}{\lambda^{s}_{M_t}} \arrow{d}[swap]{F \mu^{s,t}} & {M_s F M_t} \arrow{r}{M_s \lambda^{t}} & {M_s M_t F} \arrow{d}{\mu^{s,t}_{F}} \\
      {F M_{s \cdot t}} \arrow{rr}{\lambda^{s \cdot t}} & & {M_{s \cdot t} F}
    \end{tikzcd}
    \begin{tikzcd}
      {F} \arrow{r}{F \eta} \arrow{dr}[swap]{\eta_{F}} & {F M_e} \arrow{d}{\lambda^{e}} \\
      {} & {M_e F}
    \end{tikzcd}
  \end{equation*}
\end{define}

\label{sec:graded-em-laws}
\begin{define}
  A \intro{graded Eilenberg-Moore-law} of a \(U\)-graded monad \(P_{u} \colon \cat{C} \to \cat{C}\) over a functor \(F \colon \cat{C} \to \cat{C}\) is a family of natural transformations \(\lambda^{u} \colon P_{u}F \to FP_{u}\) indexed by the monoid \(U\) that interchanges with the graded monad multiplication and unit as below.
  \begin{equation*}
  \begin{tikzcd}
    {P_u P_v F} \arrow{r}{P_u \lambda^{v}} \arrow{d}[swap]{\mu^{u,v}_{F}} & {P_u F P_v} \arrow{r}{\lambda^{u}_{P_v}} & {F P_u P_v} \arrow{d}{F \mu^{u,v}} \\
    {P_{u \cdot v} F} \arrow{rr}{\lambda^{u \cdot v}} & & {F P_{u \cdot v}}
  \end{tikzcd}
  \begin{tikzcd}
    {F} \arrow{r}{\eta_{F}} \arrow{dr}[swap]{F \eta} & {P_e F} \arrow{d}{\lambda^{e}} \\
    {} & {F P_e}
  \end{tikzcd}
\end{equation*}
\end{define}

\begin{lemmarep}%
  \label{prop:distr-law-from-functor-distr-law}
  \kl{Graded Kleisli-laws} of a functor \(F\) over a \(T\)-graded monad \(M_{t}\) determine \kl{graded distributive laws} of the the \(\naturals\)-graded monad \(F^{n}\) over the \(T\)-graded monad \(M_{t}\).
\end{lemmarep}
\begin{proof}[Proof sketch]
  By~\cite[Lemma~6.3.7]{sokolova2005phd}, a Kleisli-law of a functor \(F\) over an ungraded monad \(M\) determines Kleisli-laws of the functors \(F^{n}\) over \(M\).
  The analogous construction in the graded case gives graded Kleisli-laws of the functors \(F^{n}\) over a graded monad \(M_{t}\).
  The other two axioms of a graded distributive law are satisfied because the graded multiplication and unit of the graded monad \(F^{n}\) are equalities and by~\cite[Lemma~6.3.8]{sokolova2005phd}.
\end{proof}
\end{toappendix}

\begin{toappendix}
\begin{lemma}
  \kl{Graded Eilenberg-Moore-laws} of a \(U\)-graded monad \(P_{u}\) over a functor \(F\) determine \kl{graded distributive laws} of the \(U\)-graded monad \(P_{u}\) over the \(\naturals\)-graded monad \(F^{n}\) arising from \(F\).
\end{lemma}
\begin{proof}
  Analogous to the Kleisli case (\Cref{prop:distr-law-from-functor-distr-law}).
\end{proof}

\noindent As a consequence of \Cref{prop:distr-law-from-functor-distr-law} and \Cref{cor:coproduct-distr-law}, we obtain \(\Samp_{T}\)-graded monads from graded Kleisli-laws and from graded Eilenberg-Moore-laws.

\begin{corollary}%
  \label{cor:samp-graded-monads}
  A \kl{graded Kleisli-law} of a functor \(F\) over a \(T\)-graded monad \(M_{t}\) gives a composite \(\Samp_{T}\)-graded monad \(M_{t_{0} \cdots t_{n}} F^{k_{0}+ \cdots + k_{n}}\).
\end{corollary}
\begin{corollary}
  A \kl{graded Eilenberg-Moore-law} of a \(U\)-graded monad \(P_{u}\) over a a functor \(F\) gives a composite \(\Samp_{U}\)-graded monad \(F^{k_{0}+ \cdots + k_{n}} P_{u_{0} \cdots u_{n}}\).
\end{corollary}
\end{toappendix}

\begin{example}%
  \label{cor:samp-graded-strong-writer}%
  \label{ex:sampling-writer}
  We will consider the \(\Samp_{T}\)-graded monad \(M_{t}(B^{k} \times -)\) arising as the composite of a constantly \(T\)-graded strong monad on a cartesian category, \(M_{t} = M\) for all \(t \in T\), and of the \((\naturals, +, 0)\)-graded writer monad, \({\writer{B}}_{k} = (B^{k} \times -)\), induced by the functor \((B \times -)\).
  The strength of a strong monad $M$ on a cartesian category induces a graded distributive law of the writer monad $\writer{B} = (B^k \times {-})$ over the constantly graded monad $M$.
  By \Cref{cor:coproduct-distr-law}, the functors \(M(B^{k_{0} + \cdots + k_{n}} \times -)\) carry a \(\Samp_{T}\)-graded monad structure.
 \end{example}

\begin{remark}
  \Cref{cor:samp-graded-strong-writer} is an instance of a more general construction, in which a \emph{functor-over-graded-monad} (or \emph{Kleisli}) distributive law $\lambda^t\colon FM_t \to M_tF$ of a functor $F$ over a graded monad $M$ induces a graded distributive law $F^nM_t\to M_tF^n$.
  An analogue construction in the inverse direction, i.e.\ \emph{graded-monad-over-functor} (or \emph{Eilenberg-Moore}) distributive laws $\lambda^t \colon M_tF \to FM_t$ determine graded monad distributive laws $F^nM_t \to M_tF^n$ (\Cref{sec:monad-functor-graded-distributive-laws}).
\end{remark}

\subsection{Feller-Dynkin processes are graded coalgebras}

We conclude the section by showing that Feller-Dynkin processes construct \(\Samp_{T}\)-graded coalgebras and that the notions of behavioural and trace equivalences that we obtain are closely related to those in the literature~\cite{chen2023behavioural}.

\begin{toappendix}
\begin{remark}%
  \label{rem:graded-kleisli-composition}
  Given a \(T\)-graded coalgebra \(\gamma_{t} \colon X \to M_{t}X\) of a \(T\)-graded monad \(M\), we can define a \intro[coalgebra composition]{composition operation} \((\ccomp)\) which is the graded version of the one for functor coalgebras~\cite[Section~6.3.2]{sokolova2005phd}: \(\gamma_{s} \ccomp \gamma_{t} = \gamma_{s} \dcomp M_{s}\gamma_{t} \dcomp \mu^{s,t}_{X}\).
  The axioms of graded coalgebras can be rephrased as \(\gamma_{s} \ccomp \gamma_{t} = \gamma_{s \cdot t}\) and \(\gamma_{e} = \id_{\ccomp} = \eta\);
  the graded coalgebra axioms, then, ensure that the composition \((\ccomp)\) is associative and unital.
  This operation is composition in the graded Kleisli category of the graded monad \(M_{t}\)~\cite[Section~3.2]{fujii2016formal}.
\end{remark}

\begin{proposition}%
  \label{prop:kleisli-labelled-are-samp}
  Consider a monoid \((T, \cdot, e)\), a \(T\)-graded monad \(M\) and an endofunctor \(F\) on a category \(\cat{C}\).
  Suppose there is a \kl{graded Kleisli-law} \(\lambda^{t} \colon FM_{t} \to M_{t}F\).
  Then, \(\Samp_{T}\)-graded coalgebras of the composite \(\Samp_{T}\)-graded monad \(M_{t_{0} \cdots t_{n}}F^{k_{0}+ \cdots + k_{n}}\) are in bijection with Kleisli-labelled coalgebras of \(M_{t}\) with labels in \(F\).
\end{proposition}
\begin{proof}
  Suppose we have a \(\Samp_{T}\)-graded coalgebra
  \[\hat{\gamma}_{(t_{0},k_{0}, \dots,t_{n}, k_{n})} \colon X \to M_{t_{0} \cdots t_{n}}F^{k_{0}+ \cdots +k_{n}}X.\]
  Then, we can consider \(\gamma_{t} = \hat{\gamma}_{(t,0)} \colon X \to M_{t}X\) and \(l = \hat{\gamma}_{(e,1)} \colon X \to M_{e}FX\).
  We check that this defines a Kleisli-labelled coalgebra using the composition defined in \Cref{rem:graded-kleisli-composition}.
  \begin{align*}
    & \gamma_{s} \ccomp \gamma_{t} & &\gamma_{e} \\
    & = \hat{\gamma}_{(s,0)} \ccomp \hat{\gamma}_{(t,0)} && = \hat{\gamma}_{(e,0)} \\
    & = \hat{\gamma}_{(s,0) \cdot (t,0)} && = \eta \\
    & = \hat{\gamma}_{(s \cdot t, 0)} &&\\
    & = \gamma_{s \cdot t} &&
  \end{align*}
  This mapping is injective because \(\Samp_{T}\)-graded coalgebras are determined by their image on their generators, \(t \in T\) and \(1 \in \naturals\):
  every component \(\hat{\gamma}_{(t_{0},k_{0}, \dots, t_{n},k_{n})}\) of a \(\Samp_{T}\)-graded coalgebra can be decomposed as
  \begin{align*}
    \hat{\gamma}_{(t_{0},k_{0}, \dots, t_{n},k_{n})} &= \hat{\gamma}_{(t_{0},0)} \ccomp \hat{\gamma}_{(e,1)}^{\ccomp k_{0}} \ccomp \cdots \ccomp \hat{\gamma}_{(t_{n},0)} \ccomp \hat{\gamma}_{(e,1)}^{\ccomp k_{n}} \\
    & = \gamma_{t_{0}} \ccomp l^{\ccomp k_{0}} \ccomp \cdots \ccomp \gamma_{t_{n}} \ccomp l^{\ccomp k_{n}}.
  \end{align*}
  Then, two \(\Samp_{T}\)-coalgebras are equal whenever all their \((t,0)\)-components and their \((e,1)\)-components are equal.

  We check that the assignment is also surjective.
  Given a Kleisli-labelled coalgebra \((\gamma_{t},l)\), we define morphisms
  \[\hat{\gamma}_{(t_{0},k_{0}, \dots, t_{n},k_{n})} = \gamma_{t_{0}} \ccomp l^{\ccomp k_{0}} \ccomp \cdots \ccomp \gamma_{t_{n}} \ccomp l^{\ccomp k_{n}}.\]
  These form a \(\Samp_{T}\)-graded coalgebra by associativity and unitality of the operation \((\ccomp)\) from \Cref{rem:graded-kleisli-composition}.
  By construction, the \((t,0)\)-components of \(\hat{\gamma}\) are precisely \(\gamma_{t}\) and its \((e,1)\)-component is \(l\).
\end{proof}
\end{toappendix}

We consider the composite \(\Samp_{T}\)-graded monad of \Cref{cor:samp-graded-strong-writer}: \kl{Feller-Dynkin processes} determine \(\Samp_{T}\)-graded coalgebras of this monad.
A similar construction has already appeared as a composition operation of functor coalgebras~\cite[Section~6.3.2]{sokolova2005phd}.

\begin{corollary}%
  \label{cor:fd-graded-coalgebra}
  A \kl{Kleisli-labelled coalgebra}, \((\gamma_{t}, \obs)\), for a constantly-graded strong monad \(M\) on a cartesian category \(\cat{C}\) and a writer functor, \((B \times -)\), determines a \(\Samp_{T}\)-graded coalgebra \(\hat{\gamma}_{(t_{0},k_{0} \dots, t_{n},k_{n})} \colon X \to M(B^{k} \times X)\) defined by induction as
  \begin{align*}
    \hat{\gamma}_{(t_{0},k_{0})} &\colon X \xrightarrow{\gamma_{t_{0}}} X \xrightarrow{\cp} X^{k_{0}+1} \xrightarrow{\obs^{k_{0}} \tensor \id} B^{k_{0}} \tensor X\\
    \hat{\gamma}_{(t_{0},k_{0} \dots, t_{n+1},k_{n+1})} & \colon X \xrightarrow{\hat{\gamma}_{(t_{0},k_{0} \dots, t_{n},k_{n})}} B^{k} \tensor X \xrightarrow{\id \tensor \hat{\gamma}_{(t_{n},k_{n})}} B^{k+k_{n+1}} \tensor X
  \end{align*}
  where composition is in the Kleisli category of \(M\), \(\cp\) denotes the diagonal morphism, and \(\obs^{k_{0}}\) denotes the \(k_{0}\)-fold monoidal product of \(\obs\) with itself in the Kleisli category of \(M\).
\end{corollary}

\begin{example}%
  \label{ex:fd-processes-coalgebra}
  A \kl{finitary Feller-Dynkin process} \((\gamma_{t}, \obs)\) determines a \(\Samp_{\posreals}\)-graded coalgebra
  \(\hat{\gamma}_{(t_{0},k_{0}, \dots, t_{n},k_{n})} \colon X \to \subdistr(B^{k} \times X)\).
  In particular, the finitary Feller-Dynkin process in \Cref{ex:fd-process-rep-system} determines the \(\Samp_{\posreals}\)-graded coalgebra in~\eqref{eq:samp-graded-rep-system}.
  Similarly, a \kl{Feller-Dynkin process} determines a \(\Samp_{\posreals}\)-graded coalgebra \(\hat{\gamma}_{(t_{0},k_{0}, \dots, t_{n},k_{n})} \colon X \to \subgiry(B^{k} \times X)\).
\end{example}

\noindent Behavioural equivalence and trace equivalence for graded coalgebras generalise bisimilarity and trace equivalence of Feller-Dynkin processes~\cite{chen2023behavioural}.
The next result relies on the characterisation of bisimilarity of \kl{Feller-Dynkin processes} as cospans of Feller-Dynkin homomorphisms~\cite[Theorem~60]{chen2023behavioural}.
Since \kl{Feller-Dynkin processes} are a subcategory of \(\Samp_{\posreals}\)-graded coalgebras for the graded monad \(\subgiry(B^{k} \times -)\) (\Cref{prop:fd-subcategory}), behavioural equivalence of \kl{Feller-Dynkin processes} may be finer than that of the corresponding graded coalgebras.

\begin{toappendix}
\begin{remark}
  Recall a consequence of Kolmogorov's extension theorem~\cite[Corollary~14.44]{klenke2008probability}.
  For every partial Markov monoid \(\gamma_{t} \colon X \to \subgiry(X)\), there is a morphism \(p^{\gamma} \colon X \to \subgiry(X^{\tensor \posreals})\) such that, for all \(t_{1}\leq \dots \leq t_{n} \in \posreals\), \(p^{\gamma} \dcomp \langle \proj[t_{1}], \dots, \proj[t_{n}] \rangle = \gamma_{t_{1}} \dcomp \langle \id, \gamma_{t_{2}-t_{1}}  \dcomp \cdots \langle \gamma_{t_{n-1}-t_{n-2}}, \gamma_{t_{n}-t_{n-1}} \rangle\rangle\).
  Note that \(X^{\tensor \posreals}\) denotes the infinite product of \(X\), which is the space of functions \(\posreals \to X\) with the product \(\sigma\)-algebra.
\end{remark}

\begin{proposition}%
  \label{prop:fd-subcategory}
  Feller-Dynkin processes and their homomorphisms~\cite[Definition~53]{chen2023behavioural} form a subcategory of \(\Samp_{\posreals}\)-graded coalgebras for the graded monad \(\subgiry(B^{k} \times -)\).
\end{proposition}
\begin{proof}
  As shown in \Cref{ex:fd-processes-coalgebra}, we can assign a graded coalgebra to each Feller-Dynkin process.
  We check that Feller-Dynkin homomorphisms are also graded coalgebra homomorphisms.
  A Feller-Dynkin homomorphism, \(h \colon (\gamma_{t}, a) \to (\delta_{t},b)\), is a morphism between their state spaces, \(h \colon X \to Y\), satisfying some continuity conditions, and that preserves the observations, \(h \dcomp b = a\), and the transitions, \(p^{\gamma} \dcomp h^{\#} = h \dcomp p^{\delta}\), where \(h^{\#} \colon X^{\tensor \posreals} \to Y^{\tensor \posreals}\) is the morphism given by precomposing with \(h\).

  Suppose \(h \colon (\gamma_{t}, a) \to (\delta_{t},b)\) is a Feller-Dynkin homomorphism.
  Then,
  \(\gamma_{t} \dcomp h = p^{\gamma} \dcomp \proj[t] \dcomp h = p^{\gamma} \dcomp h^{\#} \dcomp \proj[t] = h \dcomp p^{\delta} \dcomp \proj[t] = h \dcomp \delta_{t}\).
  This gives that \(h\) is also a homomorphism between the corresponding coalgebras.

\end{proof}
\end{toappendix}

\begin{theoremrep}%
  \label{th:fd-bisimilarity}
  If two states \(x,y\) in a Feller-Dynkin process \((\gamma_{t}, \obs)\) are bisimilar as in~\cite[Definition~23]{chen2023behavioural}, then they are behaviourally equivalent in the corresponding graded coalgebra.
\end{theoremrep}
\begin{proof}
  Two states \(x,y\) in a Feller-Dynkin process \((\gamma_{t}, \obs)\) are bisimilar if there is a cospan of Feller-Dynkin homomorphisms~\cite[Theorem~60]{chen2023behavioural}.
  By \Cref{prop:fd-subcategory}, a cospan of Feller-Dynkin processes determines a cospan of graded coalgebras.
  Therefore, if two states in a Feller-Dynkin process are bisimilar, then they must be behaviourally equivalent in the corresponding coalgebras.
\end{proof}

\noindent Intuitively, two \kl{Feller-Dynkin processes} are trace equivalent if their traces coincide when sampled a finite number of times.
The graded coalgebra associated to a \kl{monadic Feller-Dynkin process} reflects this condition: the grading monoid consists of finite lists of sampling times (\Cref{cor:fd-graded-coalgebra}).

\begin{theoremrep}%
  \label{th:fd-trace-equivalence}
  Two states \(x,y\) in a Feller-Dynkin process \((\gamma_{t}, \obs)\) are trace equivalent as in~\cite[Definition~29]{chen2023behavioural} if and only if they are trace equivalent in the corresponding graded coalgebra.
\end{theoremrep}
\begin{proof}
  Two states \(x,y \in X\) in a Feller-Dynkin process \((\gamma_{t}, \obs)\) are trace equivalent if, for all set of times \(\{t_{n} \in \posreals \mid n \in \naturals\}\) and all  \(U \subseteq A^{\tensor \posreals}\) such that \(\proj[(t_{n}, n \in \naturals)](U)\) is measurable in \(A^{\tensor \naturals}\), their traces coincide, \(p^{\gamma} \dcomp \obs^{\#} (U \mid x) = p^{\gamma} \dcomp \obs^{\#} (U \mid y) \).
  This is equivalent to \(\hat{\gamma}_{(t_{0},1,\dots,t_{n},1)} \dcomp \proj[A^{n+1}] (V \mid x) = \hat{\gamma}_{(t_{0},1,\dots,t_{n},1)} \dcomp \proj[A^{n+1}] (V \mid y)\), for all \(n \in \naturals\) and all measurable \(V \subseteq A^{n+1}\).
  This last condition is precisely trace equivalence for coalgebras (\Cref{def:trace-equivalence}).
\end{proof}

\begin{toappendix}
\begin{proposition}%
  \label{prop:em-labelled-are-samp}
  Consider a monoid \((T, \cdot, e)\), a \(T\)-graded monad \(M\) and an endofunctor \(F\) on a category \(\cat{C}\).
  Suppose there is a \kl{graded Eilenberg-Moore-law} \(\lambda^{t} \colon M_{t}F \to FM_{t}\).
  Then, \(\Samp_{T}\)-graded coalgebras of the composite \(\Samp_{T}\)-graded monad \(F^{k_{0}+ \cdots + k_{n}}M_{t_{0} \cdots t_{n}}\) are in bijection with Eilenberg-Moore-labelled coalgebras of \(M_{t}\) with \(F\)-labels.
\end{proposition}
\begin{proof}
  This proof is analogous to that of \Cref{prop:kleisli-labelled-are-samp}.
\end{proof}

\end{toappendix}

\section{Existence of a final coalgebra}
\label{sec:final-coalgebras}

Since behavioural equivalence of states is defined in terms of cospans of
coalgebra homomorphisms, a special role falls to the terminal coalgebra of any
graded monad: it is the object in which all possible behaviours of a system
materialize. The terminal morphism out of a coalgebra takes every state to its
behaviour, and two states are equivalent precisely if they are identified in
the terminal coalgebra.
The existence of terminal graded coalgebras—and of terminal non-graded
coalgebras, which are a particular case—is far from trivial~\cite{Adamek_Milius_Moss_2025}.

Terminal coalgebras are commonly constructed using \emph{terminal chains}: those
built by repeatedly applying the functor to the terminal object.\footnote{This
iteration may be carried into the transfinite domain, by taking limits for
inaccessible cardinals.} When this chain converges, the object is the carrier of the terminal coalgebra.
This approach does not directly translate to graded coalgebras: terminal
coalgebras are fixpoints of their endofunctor, but terminal graded coalgebras
are not. As such, any proof based on fixpoint iteration is bound to run into
problems.

We employ a different proof strategy: starting with an accessible monad on a
locally presentable category, we construct a series of categories, finally
arriving at the category of graded coalgebras. Since these constructions
preserve local presentability, we conclude that the category of graded
coalgebras is complete and thus has a terminal object. A similar series of constructions has
previously been used to show accessibility of the Eilenberg-Moore category of an
accessible monad \cite[Theorem 2.78]{Adamek_Rosicky_1994}.

\begin{define}[$\kappa$-accessibility]
  Let \(\kappa\) be a regular cardinal. A poset \((I, \leq)\) is \emph{\(\kappa\)-directed} if for every subset \(I' \subseteq I\) with \(|I'| \leq \kappa\), there is some upper bound \(u \in I\) such that \(i \leq u\) for all \(i\in I'\). A \emph{\(\kappa\)-directed diagram} is a functor whose domain is a \(\kappa\)-directed poset. An object \(X\) in a category \(\cat{C}\) is \(\kappa\)-presentable if the hom-functor \(\cat{C}(X,{-})\) preserves colimits of \(\kappa\)-directed diagrams.
  A category \(\cat{C}\) is called locally \emph{\(\kappa\)-accessible} if the full subcategory spanned by \(\kappa\)-presentable objects is essentially small and every object of \(\cat{C}\) is a \(\kappa\)-directed colimit of \(\kappa\)-presentable objects.
  A functor between \(\kappa\)-accessible categories is called \emph{\(\kappa\)-accessible} if it preserves \(\kappa\)-directed colimits.
  A category or functor is called \emph{accessible} if it is \(\kappa\)-accessible for some regular cardinal \(\kappa\).
  A category is \emph{locally presentable} if it is accessible and cocomplete.
  When \(\kappa = \omega\), we speak of locally finitely presentable categories and finitary functors.
\end{define}

\noindent Locally presentable categories comprise many of the categories relevant in praxis. Most relevant to our examples, the category \(\Set\) is locally finitely presentable, with a functor \(F\colon \Set\to \Set\) being finitary (\(\kappa\)-accessible) roughly if every element \(t\in FX\) mentions only finitely many (\(\kappa\)-many) elements of \(X\). In this sense, the condition of finitarity (accessibility) restricts the branching degree of the coalgebra. Other examples of locally presentable categories include the category of posets and monotone functions, categories of relational structures and relation preserving maps and metric spaces with nonexpansive maps. Also, varieties of finitary algebras, e.g. groups, lattices are locally finitely presentable. Notable non-examples are the category of topological spaces and continuous maps, as well as measurable spaces and measurable maps.

We now first construct a category of graded pre-coalgebras, essentially defined like graded coalgebras but not subject to any axioms, and carve out the category of graded coalgebras in a second step.

\begin{define}[Graded pre-coalgebra]%
  \label{def:graded-pre-coalgebra}
  Let $M$ be a $T$-graded monad on $\cat{C}$. A
  \emph{graded $M$-pre-coalgebra} consists of a $\cat{C}$-object
  $X$, and a family of $\cat{C}$ morphisms $(\gamma_t\colon X \to
  M_tX)_{t\in T}$. Morphisms in this category are defined
  as expected. We denote the category of graded
  $M$-pre-coalgebras by $\GPCoAlg{M}$, of which the
  the category of graded $M$-coalgebras
  $\GCoAlg{M}$ forms a full subcategory.
\end{define}

\begin{toappendix}  
\begin{define}[Inserter category]%
\label{def:inserter}
  Given two functors $F, G\colon \cat{C} \to \cat{D}$, the
  \emph{inserter category}, $\cat{Ins}(F, G)$, has as objects pairs $(X, f)$,
  consisting of a $\cat{C}$-object $X$ and a $\cat{D}$-morphisms
  $f\colon FX \to GX$. A homomorphism between $(X,f)$ and $(Y, g)$ is a
  $\cat{C}$-morphism $h$ such that the following square commutes:
  \[
  \begin{tikzcd}
    {FX} \arrow{r}{f} \arrow{d}{Fh} & {GX} \arrow{d}{Gh}\\
    {FY} \arrow{r}{g} & {GY}
  \end{tikzcd}
  \]
\end{define}

\begin{proposition}[]
  Let $\Delta_T\colon \cat{C} \to \prod_{t\in T}
  \cat{C}$ denote the diagonal functor.
  The category of graded $M$-pre-coalgebras is isomorphic to
  $\textbf{Ins}\left(\Delta_{T}, 
  \left(\prod_{t\in T}M_t\right) \cdot \Delta_{T}
  \right)$.
  \end{proposition} 
\begin{proof}
  Immediate from definitions.
\end{proof}
\end{toappendix}

\noindent The category $\GPCoAlg{M}$ can be represented as an inserter category, for which accessibility follows by existing results on accessible categories. We thus have the following lemma:

\begin{lemmarep}%
  \label{lem:graded-pre-coalgebras-accessible}
  Let $M$ be an accessible graded monad on an accessible category
  $\cat{C}$, i.e. a graded monad where each functor $M_t$ is accessible. Then $\GPCoAlg{M}$ is accessible.
\end{lemmarep}
\begin{proof}
  Since the subcategory $\cat{ACC}$ of $\cat{CAT}$ consisting of accessible
  categories and accessible functors is closed under products
  \cite[Proposition 2.67]{Adamek_Rosicky_1994}, we have that both
  $\Delta_{T}$ and $\left(\prod_{t\in T}M_t\right) \cdot
  \Delta_{T}$ %
  are accessible functors between accessible categories. Then, we
  have that the respective \kl{inserter category} is accessible
  \cite[Theorem 2.72]{Adamek_Rosicky_1994}, and by extension so is the
  isomorphic category of graded $M$-pre-coalgebras.
\end{proof}

\noindent In the next step, we carve out the full subcategory $\GCoAlg{M}$ of $\GPCoAlg{M}$.
The proof uses the concept of an \emph{equifier}, which allows us to encode the graded coalgebra axioms as natural transformations. It is known that the subcategory specified by an equifier inherits accessibility of the parent category \cite[Lemma 2.76]{Adamek_Rosicky_1994}.

\begin{toappendix}
  
\noindent We now show that the full subcategory of graded $M$-coalgebras is
accessible. The central tool in this step is the concept of an equifier.
\begin{define}[Equifier]%
  \label{def:equifier}
  Let $F,G\colon \cat{C} \to \cat{D}$ be functors, and
  $\phi, \psi\colon F\Rightarrow G$ natural transformations. The
  \emph{equifier} of $\phi$ and $\psi$ is the full subcategory
  $\cat{Eq}(\phi, \psi)$ of \(\cat{C}\), spanned by objects $X$ with $\phi_X = \psi_X$.
\end{define}

Equifiers %
inherit accessibility \cite[Lemma 2.76]{Adamek_Rosicky_1994}. We
use this fact in the following lemma, where we encode the graded
coalgebra axioms as equifiers of natural transformations.

\end{toappendix}
  
\begin{lemmarep}%
  \label{lem:graded-coalgebras-accessible}
  Let $M$ be an accessible graded monad on an accessible
  category $\cat{C}$. Then the category of graded
  $M$-coalgebras is accessible.
\end{lemmarep}
\begin{proof}
  Let $U \colon \GPCoAlg{M} \to \cat{C}$ denote the
  forgetful functor on graded $M$-pre-coalgebras, and for $t \in T$ let $\phi^{(t)}\colon
  U \Rightarrow M_tU$ denote the natural transformation defined at
  component $C := (X, (\alpha^{(t)}\colon X \to M_tX)_{t\in T})$ by
  $\phi_C^{(t)} = \alpha^{(t)}$. Then the unit axiom can be encoded
  via the natural transformations
  \[\eta U,\; \phi^{(e)} \colon U \Rightarrow M_eU\]
  while the multiplication axioms are encoded via
  \[(\phi^{(s)};M_s\phi^{(t)};\mu^{s,t}U),\; \phi^{(st)} \colon U \Rightarrow M_{st}U\]
  for $s, t \in T$.
  Thus, the joint equifier of these pairs of natural transformations
  (the full subcategory where all equifiers of individual pairs
  intersect) is precisely the category of graded coalgebras
  $\GCoAlg{M}$. Since $\GPCoAlg{M}$ is accessible, so
  is the joint equifier \cite[Lemma 2.76 and following remark]{Adamek_Rosicky_1994}.
\end{proof}

\noindent We lastly show that $\GCoAlg{M}$ is cocomplete, thus
rendering it locally presentable, from which completeness then
follows.

\begin{lemmarep}%
  \label{lem:forgetful-creates-colimits}
  The forgetful functor $U\colon \GCoAlg{M} \to \cat{C}$
  creates colimits.
\end{lemmarep}
\begin{proof}
  Let $D\colon I \to \GCoAlg{M}$ be a diagram and denote $Di
  = (C_i, (c^{(t)}_i)_{t\in T})$. Let $(C_i \xrightarrow{l_i}
  L)_{i\in \textsf{Ob}(I)}$ be a colimit of $UD\colon I \to
  \cat{C}$. Then, for $t\in T$, the sink
  \[(C_i \xrightarrow{c_i^{(t)}} M_tC_i \xrightarrow{M_tl_i}
    M_tL)_{i\in \mathsf{OB}(I)}\]
  is a cocone of $UD$, since for each $I$-morphism $f\colon i \to
  j$, the following diagram commutes:

  \begin{equation*}
    \begin{tikzcd}
      {X_i} \arrow{r}{c^{(t)}_i} \arrow{d}[swap]{UDf} & {M_tX_i} \arrow{d}[swap]{M_tUDf} \arrow{r}{M_tl_i} & {M_tL} \\
      {X_j} \arrow{r}[swap]{c^{(t)}_j} & {M_tX_j} \arrow{ru}[swap]{M_tl_j} & { }
    \end{tikzcd}
  \end{equation*}
  (The left square commutes since $Df$ is a homomorphism of graded
  coalgebras, the right triangle commutes since it is $M_m$
  applied to a triangle which commutes because the $l_i$ form a
  cocone.)

  This implies that there is a unique mediating morphism
  $\gamma^{(t)}\colon L \to M_tL$.
  We now have to show that
  \begin{enumerate}
    \item $(L, (\gamma^{(t)})_{t\in T})$ is a graded
          coalgebra and
    \item it is a colimit of $D$
  \end{enumerate}

  \noindent Starting with $(1)$, we need to show that the unit and
  multiplication axioms hold. For the pathological case of $I$ being
  empty, this follows since the empty colimit in $\mathcal{C}$ is
  the initial object, then the relevant diagrams have the initial
  object in the top left corner and thus commute by uniqueness of
  the outgoing morphism. If $I$ is not empty, consider the following
  diagram:

  \begin{equation*}
    \begin{tikzcd}
      {M_eC_i} \arrow[Rightarrow, no head]{r}{} & {M_eC_i} \arrow{r}{M_el_i} & {M_eL} \arrow[Rightarrow, no head]{r}{} & {M_eL} \\
      & {C_i} \arrow{lu}{\eta} \arrow{u}[swap]{c^{(e)}_i} \arrow{r}[swap]{l_{i}} & {L} \arrow{ru}[swap]{\eta} \arrow{u}{\gamma^{(e)}} & {  }
    \end{tikzcd}
  \end{equation*}

  \noindent We need to show that right triangle commutes. The outer
  path commutes by naturality of $\eta$, the left triangle because
  $Di$ satisfies the unit axiom and the middle square because
  $\gamma^{(I)}$ is defined as a colimit that requires it to
  commute. So the right triangle commutes when precomposed with $l_i$.
  Since, as a colimit cocone, the family of $l_i$ is jointly epic, the
  right hand triangle commutes. Similarly, for the multiplication
  axiom, consider the follwing diagram:

  \begin{equation*}
   \begin{tikzcd}
      {C_i} \arrow{ddd}[swap]{c_i^{(st)}} \arrow{rrr}{c^{(s)}_i} \arrow{rd}{l_{i}} & { } & & {M_sC_i} \arrow{ddd}{M_sc_i^{(t)}} \arrow{ld}[swap]{M_sl_i} \\
      { } & {L} \arrow{r}{\gamma^{(s)}} \arrow{d}[swap]{\gamma^{(st)}} & {M_sL} \arrow{d}{M_s\gamma^{(t)}} & { } \\
      { } & {M_{st}L} & {M_sM_tL} \arrow{l}{\mu^{s,t}} & { } \\
      {M_{st}C_i} \arrow{ru}[swap]{M_{st}l_i} & { } & { } & {M_sM_tC_i} \arrow{lll}{\mu^{s,t}} \arrow{lu}{M_sM_tl_i}
   \end{tikzcd}
  \end{equation*}
 
  \noindent We need to show that the inside square commutes. The outside square
  commutes since $Di$ satisfies the multiplication axiom. The top,
  left and right squares commute because $\gamma$ is a colimit that
  requires them to (or $M_m$ applied to such a square). The bottom
  square commutes because of naturality of $\mu^{s,t}$. So the inner
  square commutes when precomposed with $l_i$. Again, since the family of
  $l_i$ is jointly epic, the inner square on its own commutes.

  We now turn to $(2)$, showing that the thus defined graded coalgebra
  is a colimit of $D$. It is clear that the family $(l_i\colon C_i \to
  L)_{i\in I}$ is a cocone in $\GCoAlg{M}$, assume there is a cocone $(l'_i \colon C_i
  \to B)_{i\in I}$ of $D$ in $\GCoAlg{M}$. Since $L$ is a
  colimit in $\cat{C}$, there is a unique mediating
  $\cat{C}$-morphism $h \colon L \to B$. We now need to show that
  $h$ is a homomorphism of graded coalgebras (thus showing that such a
  unique mediating morphism also exists in $\GCoAlg{M}$).
  Consider the following diagram.

  \begin{equation*}
    \begin{tikzcd}
      {C_i} \arrow{r}{c_i^{(t)}} \arrow{d}{l_{i}} \arrow[bend right=49]{dd}[swap]{l'_i} & {M_tC_i} \arrow{d}[swap]{M_tl_i} \arrow[bend left=49]{dd}{M_tl'_i} \\
      {L} \arrow{r}{\gamma^{(t)}} \arrow{d}{h} & {M_tL} \arrow{d}[swap]{M_th} \\
      {D} \arrow{r}{d^{(t)}} & {M_tD}
    \end{tikzcd}
  \end{equation*}
  The outside square commutes (since the $l'_i$ are homomorphisms in
  $\GCoAlg{M}$), so does the top square (since the way
  $\gamma^{(m)}$ is defined forces it to commute). So the bottom
  square commutes when precomposed with $l_i$. Again we use that as a
  colimit cocone the family $(l_i \colon C_i \to L)$ is jointly epic,
  so the bottom square commutes, showing that $h$ is a homomorphism of
  graded coalgebras.
\end{proof}

\begin{toappendix}
\begin{corollary}
  \label{cor:graded-coalgebras-cocomplete}
  Let $M$ be a graded monad on $\cat{C}$. If $\;\cat{C}$ is cocomplete, then so is $\GCoAlg{M}$.
  If $M$ is accessible and $\cat{C}$ is locally
  presentable, then $\GCoAlg{M}$ is locally presentable.
\end{corollary}
\begin{proof}
  Follows immediately from the combination of
  Lemma~\ref{lem:forgetful-creates-colimits} and
  Lemma~\ref{lem:graded-coalgebras-accessible}.
\end{proof}  
\end{toappendix}

\begin{theoremrep}
  \label{thm:graded-coalgebras-complete}
  Let $M$ be an accessible graded monad on a locally
  presentable category. Then $\GCoAlg{M}$ is complete, in
  particular it has a terminal object.
\end{theoremrep}
\begin{proof}
  We know from Corollary~\ref{cor:graded-coalgebras-cocomplete} that $\GCoAlg{M}$ is locally presentable.
  It is well known that locally presentable categories are complete
  (c.f. \cite[Remark 1.56]{Adamek_Rosicky_1994}).
\end{proof}

\begin{example}
  When \(M\) has the form \(F^n\), (c.f. \Cref{exp:functor-coalgebra}) then the graded monad is accessible if and only if the functor \(F\) is accessible.
  Then \Cref{thm:graded-coalgebras-complete} recovers an established result on (ungraded) coalgebras \cite[Theorem 11.2.18]{Adamek_Milius_Moss_2025}. In this case, the terminal graded coalgebra is the one corresponding to the terminal ungraded coalgebra.
\end{example}

\begin{example}
  For the graded monad determining finitary Feller-Dynkin processes (\Cref{ex:fd-processes-labelled-coalgebra}), all components of the graded monad \(M\) are of the form \(\subdistr(B^k\times {-})\), and thus finitary since they are composed of only finitary endofunctors. We therefore have that the category \(\GCoAlg{M}\) has a terminal object.
\end{example}

\newcommand{\form}{\mathcal{F}}
\newcommand{\sem}[1]{\llbracket #1 \rrbracket}
\newcommand{\rsem}[1]{\llparenthesis #1 \rrparenthesis}
\newcommand{\by}[1]{\qquad (#1)}

\section{Characteristic Logics}%
\label{sec:characteristic-logics}
We define a notion of coalgebraic modal logic for graded coalgebras.
The core result one wants to prove about these is that they
characterise a certain notion of process equivalence—in our context,
either behavioural equivalence or trace equivalence. 
From here on forward, fix a monoid \(T\) and a \(T\)-graded monad \(M\) on \(\Set\).

\begin{define}
  A \emph{graded coalgebraic modal logic} \(\mathcal{L} = (\Theta,
  \mathcal{O}, \Lambda)\) consists of a set \(\Theta\) of constants, a
  set \(\mathcal{O}\) of propositional operators, where each \(p \in
  \mathcal{O}\) comes with an associated finite arity \(\ar(p) \in
  \naturals\), and a set \(\Lambda\) of modal operators,
  where each \(\lambda \in \Lambda\) comes with an associated finite
  arity \(\ar(\lambda) \in \naturals\), as well as a depth
  \(\depth(\lambda) \in T\). The formulae \(\form(\mathcal{L})\) of
  \(\mathcal{L}\) are then generated by the following grammar:
  \begin{equation*}
    \form(\mathcal{L}) \ni \phi_i ::= \theta \mid p(\phi_1, \ldots,
    \phi_{\ar(p)}) \mid \lambda(\phi_1, \ldots, \phi_{\ar(\lambda)}) \qquad\qquad \theta \in \Theta, p \in \mathcal{O}, \lambda
  \in \Lambda
  \end{equation*}
\end{define}

\noindent Our formulae will take values in a set of truth values \(\Omega\) (in
most instances this will instantiate to \(\Omega = 2 = \{\top, \bot\}\)).
For the semantics, we assume that each of these individual components
is equipped with a morphism in \(\Set\):
for \(\theta \in \Theta\), a morphism of type \(\hat \theta \colon 1
\to \Omega\); for propositional operators \(p \in \mathcal{O}\), a
morphism \(\sem p \colon \Omega^{\ar(p)} \to \Omega\); and for modal
operators \(\lambda \in \Lambda\), morphisms \(\sem \lambda \colon
M_{\depth(\lambda)}(\Omega^{\ar(\lambda)}) \to \Omega\).

To aid readability we restrict to unary modalities in the technical development. The extension of our results to polyadic modalities is mostly a matter of adding the appropriate indices.
Formulae \(\phi\) are interpreted in \(M\)-graded coalgebras \((X,
\gamma)\), inducing interpretation maps \(\sem \phi _\gamma \colon X
\to \Omega\), defined inductively:
for truth constants we have \(\sem\theta_\gamma = X \xrightarrow{!} 1
\xrightarrow{\hat\theta} \Omega\), for propositional operators \(\sem{ p(\phi_1, \ldots, \phi_n) }_\gamma = \langle \sem
{\phi_1}_\gamma, \ldots, \sem{\phi_n}_\gamma\rangle ; \sem p \), and for modal operators \(\sem{\lambda \phi}_\gamma = 
\gamma_{\depth(\lambda)} ; M_{\depth(\lambda)} \sem
{\phi}_\gamma ;  \sem \lambda \).

Then we say two states \(x\colon 1 \to X, y \colon 1 \to Y\) in \(M\)-graded coalgebras \((X,\gamma),
(Y, \delta)\) are \emph{logically equivalent} if for all \(\phi \in
\form(\mathcal{L})\) we have \(x;\sem \phi_\gamma =
y;\sem\phi_\delta\). Now the immediate question is: how does logical
equivalence relate to the notions of process equivalence discussed
above? In particular, we would like logical equivalence to coincide
with the process equivalence under consideration. This is known as the
\emph{Hennessy-Milner property}. In practice, the Hennessy-Milner
property is proven in two parts, showing invariance and expressivity
separately.

\begin{define}
  Let \(x \sim y\) denote either behavioural equivalence or trace
  equivalence of states \(x\) and \(y\). We say that a graded logic
  \(\mathcal{L}\) is \emph{invariant} with respect to \(\sim\), if
  \(x\) and \(y\) are logically equivalent whenever \(x \sim y\).
  Conversely, \(\mathcal{L}\) is called \emph{expressive} with respect
  to \(\sim\), if \(x\) and \(y\) being logically equivalent implies
  \(x \sim y\).
\end{define}

\noindent We will consider the cases of behavioural equivalence and trace
equivalence individually.

\subsection{Logics for Behavioural Equivalence}

The development for invariance and expressivity with respect to behavioural equivalence follows largely along the lines of the ungraded variant \cite{schroder2008expressivity}.
For behavioural equivalence, invariance of \(\mathcal{L}\) follows
for all graded modal logics, requiring no further conditions.
\begin{toappendix}
\begin{lemma}
  \label{thm:adequacy-bisim}
  Let \((X, \gamma), (Y, \delta)\) be two \(M\)-graded coalgebras and
  let \(h\colon (X, \gamma) \to (Y, \delta)\) be a homomorphism. Then
  \(\sem\phi_\gamma = h ; \sem\phi_\delta\) for all \(\phi \in
  \form(\mathcal{L})\).
\end{lemma}
\begin{proof}
  By induction on \(\phi\):
  For \(\phi = \theta \in \Theta\), this follows from the fact that
  the interpretation morphism factors through \(!\).
  For \(\phi = p(\phi_1, \ldots, \phi_n)\) we have that
  \begin{align*}
    \sem\phi_\gamma &= \langle \sem {\phi_1}_\gamma,
    \ldots, \sem{\phi_n}_\gamma\rangle ; \sem p  \\
                    &= \langle h ; \sem {\phi_1}_\delta, \ldots, h ; \sem{\phi_n}_\delta\rangle ; \sem p\\
                    &= h ; \langle \sem {\phi_1}_\delta,
                    \ldots, \sem{\phi_n}_\delta\rangle ; \sem p\\
                    &= h ; \sem \phi _\delta
  \end{align*}

  For the case \(\phi = \lambda\phi'\) we have
  \begin{align*}
    \sem\phi_\gamma &= \gamma_{\depth(\lambda)} ; M_{\depth(\lambda)}\sem {\phi'}_\gamma ; \sem \lambda\\
                    &= \gamma_{\depth(\lambda)} ; M_{\depth(\lambda)} h ; M_{\depth(\lambda)}\sem {\phi'}_\delta ; \sem \lambda \\
                    &= h ; \delta_{\depth(\lambda)} ; M_{\depth(\lambda)}\sem {\phi'}_\delta ; \sem \lambda \\
                    &= h ; \sem \phi _\delta
  \end{align*}

\end{proof}  
\end{toappendix}

\begin{theoremrep}
  \label{cor:branching-invariance}
  The logic \(\mathcal{L}\) is invariant with respect to behavioural
  equivalence.
\end{theoremrep}
\begin{proof}
  Let \(x, y\) be two behavioural equivalent states in \((X, \gamma),
  (Y, \delta)\). This implies that there are coalgebra homomorphisms
  \(g, h\) with \(g(x) = h(y)\). By \Cref{thm:adequacy-bisim} we have
  that \(x\) is logically equivalent to \(g(x)\) and \(y\) is
  logically equivalent to \(h(y)\), combined giving logical
  equivalence of \(x\) and \(y\).
\end{proof}

\noindent For expressivity it is necessary that the logic at hand contains enough modalities to observe all possible behaviours.
We call such a set of modal operators separating.

\begin{define}
  Let \(\Lambda\) be a graded set of modalities for \(M\).
  We say that \(\Lambda\) is \emph{separating} if there is a generating set $G$ of $T$ such that for all sets \(X\) and all \(g\in G\) the following source 
  is jointly injective:
  \begin{equation*}
    \{M_gX \xrightarrow{M_gf} M_g\Omega \xrightarrow{\sem{\lambda}} \Omega
    \mid f \colon X \to \Omega, \lambda \in \Lambda \text{ with } \depth(\lambda) = g\}
  \end{equation*}
\end{define}

\noindent Beyond a separating set of modalities, expressivity also requires constraints on the branching degree of the underlying system.
This is already apparent in the classical case of Hennessy-Milner logic on labelled transition systems:
expressivity in this instance only holds for finitely branching transition systems. 
Categorically, the branching degree is cast as the accessibility degree of the graded monad.
In particular, since we restrict to finitary modal logics, we restrict to \(\omega\)-accessible (i.e.\ finitary) graded monads.

\begin{theoremrep}
  \label{thm:branching-expressivity}
  Let \(M\) be a finitary graded monad and \(\mathcal{L} = (\Theta, \mathcal{O}, \Lambda)\) a graded modal logic where \(\Lambda\) is separating and \(\mathcal{O} \cup \Theta\) is functionally complete (i.e., every function \(\Omega^n \to \Omega\) arises by composing operators in \(\mathcal{O}\) and \(\Theta\)). Then \(\mathcal{L}\) is expressive for behavioural equivalence.
\end{theoremrep}
\begin{proof}
  We assume without loss of generality that all considered states are elements of the same \(M\)-coalgebra \((X, \gamma)\).
  Denote by \(R\) the logical equivalence relation on \(X\) and \(e\colon X \to X/R\) the quotient map that sends each \(z\in X\) to the equivalence class \([z] \in X/R\).
  We now construct a graded coalgebra structure \((X/R, \gamma')\), such that \(e\) is a homomorphism of graded coalgebras:
  
  We need to show that
  \begin{equation*}
    \gamma'_t(e(z)) := M_te(\gamma_t(z))
  \end{equation*}
  yields a well defined map \(\gamma'_t \colon X/R \to M_t(X/R)\), and that these maps satisfy the graded coalgebra axioms, since homomorphy already follows from definition.
  Let \(x, y\in X\) be two states in \(X\).
  We start with the case of \(t\in G\).
  For well definedness, assume \(x R y\).

  We need to show that \(M_te(\gamma_t(x)) = M_te(\gamma_t(y))\).
  Since \(\Lambda\) is separating, this follows when we show that 
  \begin{equation*}
    M_te ; M_tf ; \sem \lambda (\gamma_t(x)) =
    M_te ; M_tf ; \sem \lambda (\gamma_t(y))
  \end{equation*}
  for all \(f\colon X/R \to \Omega\) and \(\lambda \in \Lambda\) with \(\depth(\lambda) = t\).
  Since \(M_t\) is finitary, there is \(Y \subseteq X\) finite with \(\gamma_t(x), \gamma_t(y) \in M_tY \subseteq M_tX\).

  For every \(x,y\in Y\) such that not \(xRy\) we can (by definition of R) find a formula \(\phi_{x,y}\) with \(\sem{\phi_{x,y}}_\gamma(x) \not = \sem{\phi_{x,y}}_\gamma(y)\), so the semantics for all formulae is well defined and jointly injective on \(Y/R\).
  Since \(\mathcal{O}\) is functionally complete and \(Y\) is finite, we can build \(\phi\) such that \(\sem {\phi} _{\gamma \mid Y} = f_{\mid Y}\).
  Then we have that 
  \begin{alignat*}{2}
    &M_te ; M_tf; \sem {\lambda}  (\gamma_t(x))=&\\
    &M_te ; M_t\sem{\phi}_\gamma; \sem {\lambda}  (\gamma_t(x))=&\\
    & M_t\sem{\phi}_\gamma; \sem {\lambda}  (\gamma_t(x))=& \\
    &\sem{\lambda(\phi_1, \ldots, \phi_n)}_\gamma (x)= &\by{\text{def. } \sem{\cdot}_\gamma}\\
    &\sem{\lambda(\phi_1, \ldots, \phi_n)}_\gamma (y)=&\by{x R y}\\
    &M_t\sem{\phi}_\gamma ;\sem {\lambda}( \gamma_t(y))=&\by{\text{def. } \sem{\cdot}_\gamma}\\
    &M_te ; M_t\sem{\phi}_\gamma ; \sem {\lambda}  (\gamma_t(y))=&\\
    &M_te ; M_tf; \sem {\lambda}  (\gamma_t(y))=&\\
  \end{alignat*}
  Now we show well definedness and the graded coalgebra axioms at once. Assume that \(t\in T\setminus G\), then since \(G\) is generating, \(t= t_1 \ldots t_m\) with \(t_i \in G\).
  Now we show well definedness for \(t\) by indution over \(m\), so let \(t' = t_1\ldots t_{m-1}\).
  Well definedness then follows from the following diagram: 
\begin{equation*}
  \begin{tikzcd}
    X \arrow{r}{"\gamma_{t'}"} \arrow{d}{"e"} \arrow[bend left]{rrr}{"\gamma_t"} & M_{t'}X \arrow{r}{"M_{t'}\gamma_{t_m}"} \arrow{d}{"M_{t'}e"} & M_{t'}M_{t_m}X \arrow{r}{"{\mu^{t',t_m}}"} \arrow{d}{"M_{t'}M_{t_m}e"} & M_tX \arrow{d}{"M_te"}\\
    X/R \arrow{r}{"\gamma'_{t'}"} \arrow[bend right]{rrr}{"\gamma'_t"}              & M_{t'}(X/R) \arrow{r}{"M_{t'}\gamma'_{t_m}"}                  & M_{t'}M_{t_m}(X/R) \arrow{r}{"{\mu^{t'a,t_m}}"}                         & M_t(X/R)              
  \end{tikzcd}
\end{equation*}
The morphism \( \gamma'_{t'} ; M_{t'}\gamma'_{t_m} ; \mu^{t',t_m} \) is well defined (because all constituent morphisms are by induction). The top commutes by the graded coalgebra axioms, left and middle square commute by definition of \(\gamma'\) and the right square commutes by naturality of \(\mu\). Then we have that
\begin{align*}
  \gamma'_t(e(z)) &= M_te(\gamma_t(z))\\
  &=\gamma_{t'}; M_{t'}\gamma_{t_m}; \mu^{t',t_m}; M_te (z)\\
  &=e; \gamma'_{t'}; M_{t'}\gamma'_{t_m} ;\mu^{t',t_m} (z)
\end{align*}
Since \(e\) is an epimorphism, we have that \(\gamma'_t = \gamma'_{t'} ; M_{t'}\gamma'_{t_m} ; \mu^{t',t_m}\), yielding well definedness as well as the second graded coalgebra axiom.
A similar argument can be made for any choice of \(t',t_m\in T\), showing that the second graded coalgebra axiom holds for all indices.
For the first graded coalgebra axiom the statement follows by application of naturality.
\end{proof}

\begin{example}
  When instantiating to functor coalgebras, i.e. graded coalgebras for graded monads of the form \(F^n\) as in \Cref{exp:functor-coalgebra}, considering the truth value object \(2\) and propositional operators \(\top, \neg, \land\) with their usual semantics, then \Cref{thm:branching-expressivity} instantiates to (a finitary version of) known expressivity results for coalgebraic modal logic \cite{schroder2008expressivity}.
\end{example}

\begin{example}
\label{exp:prob-branching-logic}
  Consider finitary Feller-Dynkin Processes as in \Cref{ex:fd-processes-coalgebra}, i.e. graded coalgebras of the \(\Samp_{\posreals}\)-graded monad \(M\) where the components \(M_{(t_0,k_0,\ldots,t_n,k_n)}\) are of the form \(\subdistr(B^k \times {-})\), with \(k = t_1 +\dots+t_n\).
  We define a logic \(\mathcal{L} = (\Theta, \mathcal{O}, \Lambda)\), with \(\Theta = \{\top\}\), \(\mathcal{O} = \{\land, \neg\}\) and \(\Lambda = \{\rmodal{b}_p, \modal{r}_p \mid b\in B, r\in \posreals, p\in [0,1]\}\).
  The semantics of this logic is defined on the truth value object \(\Omega = 2\), with \(\top, \land\) and \(\neg\) having the usual interpretation.
  For the modal operators we have depths \(\depth(\rmodal{b}_p) = (0,1)\) and \(\depth(\modal{r}_p) = (r,0)\), with \(\sem{\rmodal{b}_p} \colon M_{(0,1)}\Omega \to \Omega (= \mathcal{D}(B\times 2) \to 2)\) given by \(\sem{\rmodal{b}_p}(\mu) = \top \) iff \(\mu(b,\top) \geq p\), and similarly \(\sem{\modal{r}_p} \colon M_{(r,0)}\Omega \to \Omega (= \mathcal{D}2 \to 2)\) given by \(\sem{\rmodal{r}_p}(\mu) = \top \) iff \(\mu(\top) \geq p\).
  This set of modalitites is \(G\)-separating, and thus we have invariance and expressivity of \(\mathcal{L}\) for behavioural equivalence by \Cref{cor:branching-invariance} and \Cref{thm:branching-expressivity} respectively. 
  \begin{toappendix}
    \noindent \textbf{Details for \Cref{exp:prob-branching-logic}}\\
    We show that the set of modalities in \Cref{exp:prob-branching-logic} is separating.
    For \(g = (r,0)\), let \(\mu,\nu\in \subdistr X\), such that \(\subdistr f;\sem{\modal{r}_p}(\mu) = \subdistr f;\sem{\modal{r}_p}(\nu)\) for all \(f\colon X \to 2\) and \(p \in [0,1]\). We have to show that \(\mu = \nu\).
    Let \(\mu(x) = q\), and \(f_x \colon X \to 2\) the characteristic function of \(x\).
    Then \(\subdistr f_x;\sem{\modal{r}_q}(\nu) = \top\) iff \(\nu(x) \geq q =\mu(x)\). Similarly we can show that \(\mu(x)\leq \nu(x)\) and thus \(\mu(x) = \nu(x)\) for all \(x\in X\).
    When \(g = (0,1)\), let \(\mu, \nu \in \subdistr(B\times X)\), such that  \(\subdistr f;\sem{\rmodal{b}_p}(\mu) = \subdistr f;\sem{\rmodal{b}_p}(\nu)\) for all \(f\colon X \to 2\) and \(p \in [0,1]\). We have to show that \(\mu = \nu\). Let \(\mu(b,x) = q\), and \(f_x \colon X \to 2\) the characteristic function of \(x\).
    \(\subdistr f_x;\sem{\rmodal{b}_q}(\nu) = \top\) iff \(\nu(b,x) \geq q =\mu(b,x)\). Similarly we can show that \(\mu(x)\leq \nu(x)\) and thus \(\mu(x) = \nu(x)\) for all \(x\in X\).
  \end{toappendix}
\end{example}

\subsection{Logics for Trace Equivalence}%
\label{sec:trace-logics}

\noindent As our notion of trace equivalence is inspired by graded semantics, the development of characteristic logics builds on the respective notion of \emph{graded logics} \cite{milius_et_al:LIPIcs.CALCO.2015.253,DorschMS19,ForsterSWBGM24}.
As such, we require an additional condition on the graded monad at hand. In the context of trace semantics, we assume the generating set $G$ of $T$ is fixed from the outset, and all modal operators $\lambda\in \Lambda$ have a depth $\depth(\lambda)\in G$.

\begin{define}%
  \label{def:G-uniform}
  A \(T\)-graded monad \(M\) is called \emph{\(G\)-uniform}, if the following diagram is a coequalizer in the Eilenberg-Moore category of \((M_e, \mu^{e,e}, \eta)\) for all \(g\in G\) and \(t\in T\).
  \begin{equation*}
    \begin{tikzcd}
      M_gM_eM_t \arrow[shift left=0.8ex]{r}{"M_g\mu^{e,t}"} \arrow[shift
      right=0.8ex]{r}[swap]{"\mu^{g,e}M_t"} & M_gM_t
      \arrow{r}{"\mu^{g,t}"} & M_{g\cdot t}
    \end{tikzcd}
  \end{equation*}
\end{define}

\noindent Commutativity of the diagram follows from the graded monad
axioms, so only universality is additionally required for
uniformity. The intuition of $G$-uniformity is that behaviours of different depths do not bleed into each other when composed, which will allow us to reverse the composition and evaluate arguments of modalities on the relevant component behaviour.

\begin{example}%
  \label{exp:g-uniform-monads}
  \begin{enumerate}[wide]
    \item The \(\naturals\)-graded monad determined by an endofunctor \(F\), with \(M_n = F^n\) (c.f. \Cref{prop:functor-coalgebras}) is \(G\)-uniform for \(G = \{1\}\).
    \item In fact, the condition from graded semantics of a graded monad being \emph{depth-1}~\cite{DorschMS19} is precisely an instance of \(G\)-uni\-for\-mi\-ty, where the monad is \(\naturals\)-graded and \(G = \{1\}\). As such, all depth-1 graded monads are \(G\)-uniform.
    \item The \(\Samp_{T}\)-graded monad of the form \(M_{(t_0,k_0,\ldots,t_n,k_n)} = M(B^n \times {-})\) as in \Cref{cor:fd-graded-coalgebra} is \(G\)-uniform for $G = \{(r,0), (0,1) \mid r \in \posreals\}$
  \end{enumerate}
\end{example}

\begin{toappendix}
  \noindent \textbf{Datails for \Cref{exp:g-uniform-monads}}
  We show explicitly that the \(\Samp_{T}\)-graded monad of the form \(M_{(t_0,k_0,\ldots,t_n,k_n)} = M(B^n \times {-})\) as in \Cref{cor:fd-graded-coalgebra} is \(G\)-uniform, where \(G = \{(g,0),(e,1) \mid g\in G'\}\) for a generating set \(G'\) of \(G\).
  
  Let \(\tau\) be the strenght of the monad \(M\)
  We treat the case of generator \((g,0)\) explicitly, the case of \((e,1)\) is
  analougous and simpler.
  We have to show that the following diagram is a coequalizer:
  \begin{equation*}
    \begin{tikzcd}
      M(B \times MM(B^n \times {-}) \arrow[shift left=0.8ex]{r}{"M(B \times \mu)"} \arrow[shift
      right=0.8ex]{r}[swap]{"M\tau ; \mu"} & M(B\times M(B^n \times {-}))
      \arrow{r}{"M\tau ; \mu"} & M(B^{n+1} \times {-})
    \end{tikzcd}
  \end{equation*}
  We show that this is a split coequalizer. The splittings we investigate are
  \[M(B\times \eta (B^n \times {-})) \colon M(B^{n+1}\times {-})\to M(B\times M(B^n \times {-}))\]
  and
  \[M(B\times M\eta(B^n \times {-})) \colon M(B\times M(B^n \times {-})) \to M(B\times MM(B^n \times {-}))\]

  For these to give a split coequalizer, we need to verify the follwoing identities:
  \[M(B\times \eta (B^n \times {-})); M\tau ; \mu = \id\]
  which follows by the axioms of a (strong) monad and

  \[M(B\times M\eta(A^n \times {-}));M(B \times \mu) = \id\]
  which is also by the monad axioms.

  The third equality we need to verify is given by the following square:

  \begin{equation*}
    \begin{tikzcd}[column sep = huge]
      M(B\times M(B^n \times {-})) \arrow{r}{M(B\times M\eta(A^n \times {-}))} \arrow{d}{M\tau; \mu} & M(B\times MM(B^n \times {-})) \arrow{d}{M\tau; \mu}\\
      M(B^{n+1} \times {-}) \arrow{r}{M(B\times \eta(B^n \times {-}))} & M(B\times M(B^n \times {-}))
    \end{tikzcd}
  \end{equation*}
  Which commutes by naturality of \(\tau\) and \(\mu\).
\end{toappendix}

\noindent In the case of trace equivalence, invariance is no longer a given: we generalize
the concept of \emph{graded logics} from graded semantics, which puts restrictions on the admissible operators, to the present context
to ensure this property.
Furthermore, the separatedness proof for trace equivalence can no
longer rely on the set of propositional operators being functionally complete.
Thus, our expressivity criterium leaves more work to the concrete instantiation
than the criterium for behavioural equivalence.

\begin{define}
  The logic \(\mathcal{L} = (\Theta, \mathcal{O}, \Lambda)\) is a \emph{trace-logic} if:
  \begin{enumerate*}[label=(\roman*)]
    \item \(\Omega\) is equipped with an \(M_e\)-algebra structure \(o\colon M_e
    \Omega \to \Omega\);
    \item the evaluation morphisms \(\sem p \colon \Omega^n \to \Omega\) of
    propositional operators \(p\in \mathcal{O}\) are algebra homomorphisms
    \(\sem p \colon (\Omega, o)^n \to (\Omega, o)\);
    \item the evaluation morphisms \(\sem \lambda \colon M_t \Omega \to
    \Omega\) makes the following two diagrams commute:\\
    \begin{equation*}
      \begin{tikzcd}
        & M_tM_e\Omega \arrow{r}{"\mu^{t,e}"} \arrow{d}{"M_to"} & M_t\Omega \arrow{d}{"\sem \lambda"}\\
        & M_t\Omega \arrow{r}{"\sem\lambda"} & \Omega
      \end{tikzcd}
      \begin{tikzcd}
        & M_eM_t\Omega \arrow{d}{"\mu^{e,t}"} \arrow{r}{"M_e\sem \lambda"} & M_e\Omega \arrow{d}{"o"}\\
        & M_t\Omega \arrow{r}{"\sem\lambda"} & \Omega
      \end{tikzcd}
    \end{equation*}
  \end{enumerate*}
\end{define}

\noindent The right-hand diagram says that the semantics of modal operators are algebra homomorphisms, while the diagram on the left will allow us to construct a competitor to the coequalizer in \Cref{def:G-uniform}. If a logic satisfies the above requirements, we get invariance for the fragment of uniform-depth
formulae, in which constants are restricted to occur at the same depth in
the syntax tree.

\begin{define}[Formulae of uniform depth]
  For $k \in T$, we write $\form_k(\mathcal{L})$ to denote the set of $\mathcal{L}$ formulae of uniform depth \(k\), which are inductively defined by the following grammars:
  \begin{alignat*}{2}
    &\form_e(\mathcal{L}) \ni \phi := \theta \mid \lambda(\phi) \mid p(\phi_1, \ldots, \phi_{\ar(p)})\qquad &\text{ for $k = e$}\\
    &\form_k(\mathcal{L}) \ni \phi := \lambda(\phi) \mid p(\phi_1, \ldots, \phi_{\ar(p)}) \qquad &\text{ for $k \not = e$}
  \end{alignat*}
  where $\theta \in \Theta$, $\lambda \in \Lambda$ and $p\in \mathcal{O}$, and all \(\phi_i \in \form_l(\mathcal{L})\) have equal uniform depth \(l\), with \(l\) subject to the constraints \(l = k\) in the case of propositional operators and \( \depth(\lambda) \cdot l  = k\) in the case of a modal operators.
\end{define}

\noindent To prove invariance for trace semantics, we define a semantics \(\rsem\phi\colon M_k1 \to \Omega\) which operates directly on the observed behaviours and then show that \(\sem \phi_\gamma\) factors through \(\rsem \phi\).

\begin{define}
  For \(\phi\in \form_k(\mathcal{L})\), define the homomorphism of \(M_e\)-algebras \(\rsem \phi \colon (M_k1, \mu^{e,k}) \to (\Omega, o)\) inductively: For truth constants we define \(\rsem \theta = M_e\hat \theta ; o\) and for propositional operators \(\rsem{p(\phi_1, \ldots, \phi_n)} = \langle \rsem \phi_1, \ldots, \rsem{\phi_n} \rangle; \sem{p}\). For formulae of the form \(\lambda\phi\) consider the following diagram:
  \begin{equation}
    \tag{$\ast$}
    \label{eq:semantics}
    \begin{tikzcd}[column sep = 2cm, row sep = 0.7cm]
      M_{\depth(\lambda)}M_eM_l1 \arrow[shift left]{r}{\mu^{\depth(\lambda),e}} \arrow[shift right,swap]{r}{M_{\depth(\lambda)}\mu^{e,l}} \arrow[swap]{d}{"M_{\depth(\lambda)}M_e\rsem{\phi}"} & M_{\depth(\lambda)}M_l1 \arrow{r}{\mu^{\depth(\lambda),l}} \arrow{d}{"M_{\depth(\lambda)}\rsem{\phi}"} & M_k1 \arrow[dashed]{d}{"\rsem{\lambda\phi}"}\\
      M_{\depth(\lambda)}M_e\Omega \arrow[shift left]{r}{\mu^{\depth(\lambda),e}} \arrow[shift right,swap]{r}{M_{\depth(\lambda)} o} & M_{d(\lambda)}\Omega \arrow{r}{\sem \lambda} & \Omega
    \end{tikzcd}
  \end{equation}

  \noindent The morphisms on the top are the coequalizer diagram of \Cref{def:G-uniform}, the bottom morphisms commute due to the second modal operator axiom. By naturality/homomorphy, we then have that \( M_{\depth(\lambda)}\rsem{\phi} ; \sem \lambda\) is a competitor to the coequalizer \(\mu^{\depth(\lambda),l}\). We then define \(\rsem{\lambda\phi}\colon M_k1 \to \Omega\) to be the unique morphism that makes the right hand square commute.
\end{define}

\noindent We are now able to prove invariance for trace logics.
\begin{toappendix}
  \begin{lemma}%
  \label{lem:semantics-factors}
  Let \(\phi \in \mathcal{L}_k\) be a uniform depth \(k\) formula in a trace logic \(\mathcal{L} = (\Theta, \mathcal{O}, \Lambda)\) for a \(G\)-uniform graded monad \(\mathbb{M}\). Let \((X, \gamma)\) be a graded coalgebra for \(\mathbb{M}\). Then \(\sem \phi _\gamma = \gamma_k; M_k! ;\rsem \phi\).
\end{lemma}
\begin{proof}
  We show the claim by structural induction over \(\phi\).
  For the case of \(\phi = \theta \in \Theta\), we have that 
  \begin{alignat*}{2}
    \sem \theta _\gamma &= !; \hat \theta &\\
    &= !; \hat\theta ; \eta_\Omega; o &\by{o\text{ }M_0\text{-Alg.}}\\
    &= \eta_X; M_0!; M_e\hat\theta; o &\by{\eta \text{ nat.}}\\
    &= \gamma_e ; M_e!; \rsem{\theta} &\\
  \end{alignat*}

  For the case of \(\phi = p(\phi_1, \ldots, \phi_n)\), we have that
  \begin{alignat*}{2}
    \sem{p(\phi_1, \ldots, \phi_n)}_\gamma &= \langle\sem{\phi_1}_\gamma, \ldots, \sem{\phi_n}_\gamma\rangle;\sem{p} &\\
    &= \langle\gamma_k ; M_k! ;\rsem{\phi_1}, \ldots, \gamma_k ; M_k! \rsem{\phi_n} \rangle;\sem{p}&\\
    &= \gamma_k ;M_k! ; \langle\rsem{\phi_1}, \ldots, \rsem{\phi_n}\rangle;\sem{p} &\\
    &= \rsem{p(\phi_1, \ldots, \phi_n)}.
  \end{alignat*}

  And, for the case of \(\phi = \lambda\phi'\), where \(\phi_i\in \mathcal{L}_l\), we have
  \begin{alignat*}{2}
    \sem{\lambda\phi'}_\gamma &= \gamma_l ; M_{\depth(\lambda)}\sem{\phi'}_\gamma ;  \sem{\lambda} &\\
    &= \gamma_{\depth(\lambda)} ; M_{\depth(\lambda)}( \gamma_l; M_l!; \rsem{\phi'}) ; \sem{\lambda} &\\
    &=\gamma_{\depth(\lambda)}; M_{d(\lambda)}(\gamma_l;M_l!);  \mu^{d(\lambda),l}; \rsem{\lambda\phi'}  &\\
    &= \gamma_{\depth(\lambda)}; M_{d(\lambda)} \gamma_l ;  \mu^{d(\lambda),l} ; M_{d(\lambda)l}!; \rsem{\lambda\phi'} &\\
    &= \gamma_{d(\lambda),l}; M_{d(\lambda)l}!; \rsem{\lambda\phi'} \qedhere
  \end{alignat*}
\end{proof}
\end{toappendix}

\begin{theorem}%
  \label{cor:trace-invariance}
  Let \(\mathcal{L} = (\Theta, \mathcal{O}, \Lambda)\) be a trace logic for a \(G\)-uniform graded monad \(\mathbb{M}\). Then the uniform-depth fragment of \(\mathcal{L}\) is invariant with respect to trace equivalence in graded coalgebras for \(\mathbb{M}\).
\end{theorem}

\noindent We next turn our attention to expressivity for trace semantics.

\begin{define}
    The trace logic \(\mathcal{L}\) is \emph{unit separating}, if the set of morphisms \(\{\rsem\theta , \rsem{p} \colon M_e1 \to \Omega\mid \theta \in \Theta, p\in \mathcal{O}, \ar(p)=0\} \) is jointly injective.
    Moreover, \(\mathcal{L}\) is \emph{inductively separating}, if  for all \(g \in G, t\in T\) and jointly injective sets of algebra morphisms \(\mathfrak{A} \subseteq (M_t1, \mu^{e,t}) \to (\Omega, o)\) closed under \(\mathcal{O}\), the set of morphisms
    \(\{\rsem{\lambda h} \colon M_{g\cdot t}1 \to \Omega\mid \lambda\in \Lambda, \depth(\lambda) = g, h\in \mathfrak{A}\}\)
    is jointly injective, where \(\rsem{\lambda h}\) is defined analogously to the dashed arrow in diagram (\ref{eq:semantics}).
\end{define}

\begin{theoremrep}%
  \label{thm:trace-expressivity}
  When \(M\) is a graded monad over \((T, \cdot,e)\), and \(\mathcal{L} = (\Theta, \mathcal{O}, \Lambda)\) is a trace logic which is unit separating and inductively separating, then the uniform-depth fragment of \(\mathcal{L}\) is expressive with respect to trace equivalence.
\end{theoremrep}
\begin{proof}
  For \(t\in T\), the set of morphisms \(\{\rsem \phi \colon M_t1\to \Omega \mid \phi \in \mathcal{L}_t \}\) is jointly injective. The claim then follows by application of \Cref{lem:semantics-factors}. Since \(G\) is a generator of \(T\), we can write \(t = g_1\cdot \ldots\cdot g_n\) and proceed by induction over \(n\).
  The case for \(n=0\), i.e.\ \(t=e\), follows by unit separation; for \(t\not = e\), we have by induction hypothesis that the set of evaluation morphisms \(\{\rsem{\phi}\colon M_{t'}1\to \Omega \mid \phi \in \mathcal{L}_{t'}\}\) for \(t' = g_2 \cdot \ldots \cdot g_n\) is jointly initial and closed under \(\mathcal{O}\).
  Then we have, since \(\mathcal{L}\) is inductively separating, that the set \(\{\rsem{\lambda \phi}\colon M_{t}1 \to \Omega \mid \lambda \in \Lambda, \depth(\lambda) = g_1, \phi \in \mathcal{L}_{t'} \}\subseteq \{\rsem \phi \mid \phi \in \mathcal{L}_{t}\}\) is jointly injective.
\end{proof}

\begin{example}
  When instantiating to graded semantics (\Cref{sec:graded-semantics}), the axioms for trace-logics are equivalent to graded logics~\cite{DorschMS19}, and by application of \Cref{cor:trace-invariance} and \Cref{thm:trace-expressivity} we recover the relevant invariance and expressivity results~\cite{ForsterSWBGM24}.
\end{example}

\begin{example}%
  \label{exp:prob-trace-logic}
  Consider the setting of \Cref{ex:fd-processes-coalgebra}, with distributions instead of subdistributions, that is, we have a \(\Samp_{\posreals}\) graded monad \(M\) where each \(M_t\) is of the form \(\distr(B^k \times {-})\). This modification is without loss of generality, since we can always add a sink state to systems.
  We define a logic \(\mathcal{L} = \{\Theta, \mathcal{O}, \Lambda\}\) with \(\Theta = \emptyset\), \(\mathcal{O} = \{+_p, \top \mid p \in [0,1]\}\), where \(\ar(\top) =0, \ar(+_p) = 2\) and unary modalities \(\Lambda = \{\modal{r}, \rmodal{b} \mid r \in \posreals, b \in B\}\) with \(\depth(\modal{r}) = (r, 0)\) and \(\rmodal{b} = (0, 1)\).
  As a truth value object we choose \(\Omega = [0,1] \subseteq \reals\) and $\distr$-algebra structure $o$ taking expected values: \(o(\mu) = \sum_{x\in \Omega } \mu(x)\). We assign semantics to these operators via \(\sem {+_p}\colon \Omega^2 \to \Omega\) calculating weighted sums: \(\sem{+_p}(a,b) = pa + (1-p)b\), the semantics of the 0-ary propositional operator \(\top\) is the constant function \(1\), the operator \(\sem{\modal{r}} \colon M_{(r,0)}\Omega \to \Omega \; ( = \distr[0,1] \to [0,1]\)) takes expected values and \(\sem{\rmodal b}\colon M_{(0,1)}\Omega \to \Omega\; ( = \distr(B \times [0,1]) \to [0,1])\) is calculated via $\sem {\rmodal b} (\mu) = \sum_{v\in [0,1]} \mu(b,v)v$.
  Intuitively, the modal operator \(\modal{r}\) progress the system by time \(r\), while \(\rmodal b\) probes whether the current state is labelled by \(b\). The semantics then gives the expected value of a formula holding when the system is probabilistically executed.
  Simple calculation shows that these operations satisfy the axioms of trace logics. Therefore, by \Cref{cor:trace-invariance}, \(\mathcal{L}\) defines a logic that is invariant for trace semantics of Feller-Dynkin processes.
  The logic also satisfies unit separation and inductive separation, and thus we have by \Cref{thm:trace-expressivity} that \(\mathcal{L}\) is expressive for trace semantics.
  This logic is a continuous time version of the multi-valued modal logic characterising trace semantics in probabilistic transition systems \cite{BernardoB08}.
\end{example}

\begin{toappendix}
   \noindent \textbf{Details for \Cref{exp:prob-trace-logic}}\\
   We state the argument for expressivity explicitly.
   It is immediate that \(\rsem{\top} \colon M_{(0,0)}1 = \distr 1 \cong 1\to \Omega\) is injective, giving unit separation. For depth-1 separation, assume \(\mathfrak{A}\subseteq M_{t'}1  \to \Omega (= \distr(B^n) \to [0,1])\) is a set of affine maps, closed under affine maps itself. We distinguish the case where the generator \(g = (r,0)\) and \(g = (0,1)\). For \(g = (r,0)\), assume \(\mu, \nu \in \distr(B^n)\) such that \(\rsem{\modal r f}(\mu) = \rsem{\modal r f}(\nu) \) for all \(f\in \mathfrak{A}\). We have to show that \(\mu = \nu\). Note that
  \begin{align*}
    & f(\mu)\\
    &= \distr f (1\cdot \ket{\mu}); \sem{\modal r}\\
    &= \rsem{\modal r f}(\mu) \\
    &= \rsem{\modal r f}(\nu) \\
    &= \distr f (1\cdot \ket{\nu}); \sem{\modal r}\\
    &= f(\nu)
  \end{align*}
  Then \(\mu = \nu\) follows by joint injectivity of \(\mathfrak{A}\).
  For \(g=(0,1)\), assume again that \(\mu, \nu \in \distr(B^{n+1})\) such that \(\rsem{\rmodal{b}f}(\mu) = \rsem{\rmodal{b}f}(\nu)\) for all \(b\in B\) and \(f \in \mathfrak{A}\).
  When we choose \(f\) to be the constant function \(1\), then we get from this assumption that
  \begin{align*}
    p_b:&=\sum_{c\in B^n} \mu(b, c)\\
    &=\sum_{c\in B^n} \mu(b, c)f(1 \cdot \ket{c})\\
    &=\rsem{\rmodal{b}f}(\mu)\\
    &= \rsem{\rmodal{b}f}(\nu)\\
    &=\sum_{c\in B^n} \nu(b, c)f(1 \cdot \ket{c})\\
    &=\sum_{c\in B^n} \nu(b, c)
  \end{align*}
  We are then able to normalize the distributions by choosing \(\mu',\nu' \in \distr(B\times \distr(B^n)) = M_{(0,1)}M_{t'}1\) as follows:
  \[\mu' := \sum_{b\in B} p_b \ket{\mu_b}\]
  where
  \[\mu_b := \sum_{c\in B^n} \frac{\mu(b,c)}{p_b}\ket{c}\]
  (analoguously for \(\nu'\)). Then we have that
  \begin{align*}
    & f(\mu_b)p_b\\
    &= \distr(B\times f) (\mu'); \sem{\rmodal b}\\
    &= \rsem{\rmodal b f}(\mu) \\
    &= \rsem{\rmodal b f}(\nu) \\
    &= \distr(B\times f) (\nu'); \sem{\rmodal b}\\
    &= f(\nu_b)p_b
  \end{align*}
  When \(p_b \not = 0\) then we have by joint injectivity that \(\mu_b=\nu_b\). Since the case of \(p_b = 0\) is irrelevant, we have that \(\mu = \nu\).
  Thus we have that \(\mathcal{L}\) is an expressive multi-valued logic for trace semantics.
\end{toappendix}

\begin{toappendix}
\end{toappendix}

\section{Conclusion and Future Work}

We introduced graded coalgebras of graded monads, and their behavioural and trace equivalences; we argued that graded monad coalgebras model state-based systems which exhibit continuous-time behaviour.
We proved that graded monads can be combined via graded distributive laws.
As a particular case, \(T\)-graded monads regulating the branching behaviour can be combined with functors regulating the observable behaviour;
their composition gives a monad graded by the coproduct monoid \(T+\naturals\).
Graded coalgebras of these \(T+\naturals\)-graded monads instantiate to Feller-Dynkin processes~\cite{chen2023behavioural} and capture their bisimilarity and trace equivalence.
We developed the theory of graded monad coalgebras.
We proved existence of terminal graded monad coalgebras, under suitable assumptions.
We defined characteristic coalgebraic modal logics for behavioural and trace equivalence, and proved invariance and expressivity for these logics.

While the present work focuses on probabilistic processes as central examples, we anticipate that the technique of externalising the time parameter could have wider applications and facilitate the coalgebraic treatment of, for instance,  timed automata and related systems.
Furthermore, it will be interesting to use the framework we laid out to connect recent work on behavioural pseudometrics for continuous-time Markov processes~\cite{ChenCP25} with behavioural metrics in coalgebra~\cite{BaldanBKK18}.
Further work could look at enriching graded monad coalgebras to impose continuity conditions on the dependency of the transitions with respect to the grading parameter.

\bibliographystyle{plainurl}
\bibliography{main.bib}

\end{document}